\documentclass{article}

\usepackage{arxiv}

\usepackage[utf8]{inputenc} % allow utf-8 input
\usepackage[T1]{fontenc}    % use 8-bit T1 fonts
\usepackage{hyperref}       % hyperlinks
\usepackage{url}            % simple URL typesetting
\usepackage{enumitem}
\usepackage{bbm}
\usepackage{verbatim}
\usepackage{subfig}
\usepackage{tikz}
\usepackage{booktabs}       % professional-quality tables
\usepackage{amsfonts}       % blackboard math symbols
\usepackage{nicefrac}       % compact symbols for 1/2, etc.
\usepackage{mathtools}
\usepackage{microtype}      % microtypography
\usepackage{lipsum}
\usepackage{amsmath}
\usepackage{amsthm}

\usepackage{bbm}
\usepackage{authblk}

\usepackage{amssymb}
\usepackage{mathtools}
\usepackage{graphicx}
\newtheorem{theorem}{Theorem}[section]

\newtheorem{lemma}[theorem]{Lemma}
\newcommand{\Tr}{\text{Tr}}
\newcommand{\trg}{\text{trg}}

\graphicspath{ {./images/} }

\title{A spin glass model for the loss surfaces of generative adversarial networks}

\usepackage[backend=biber,style=alphabetic]{biblatex}    
\usepackage{bm}
\author[1]{\large Nicholas P. Baskerville}
\author[2]{\large Jonathan P. Keating}
\author[1]{\large Francesco Mezzadri}
\author[1]{\large Joseph Najnudel}

\affil[1]{\small\textit{School of Mathematics, University of Bristol, Fry Building, Bristol, BS8 1UG, UK}}
\affil[2]{\small\textit{Mathematical Institute, University of Oxford, Oxford, OX2 6GG, UK}}
\affil[ ]{\texttt{\{n.p.baskerville, F.Mezzadri, joseph.najnudel\}@bristol.ac.uk}, \texttt{Jon.Keating@maths.ox.ac.uk}}

\newcommand{\wD}{w^{(D)}}
\newcommand{\vwD}{\vec{w}^{(D)}}
\newcommand{\wG}{w^{(G)}}
\newcommand{\vwG}{\vec{w}^{(G)}}
\newcommand{\lD}{\ell^{(D)}}
\newcommand{\lG}{\ell^{(G)}}
\newcommand{\LD}{L^{(D)}}
\newcommand{\LG}{L^{(G)}}
\newcommand{\expect}{\mathbb{E}}
\newcommand{\indic}{\mathbbm{1}}
\renewcommand{\vec}[1]{\bm{#1}}
\newcommand{\pD}{\partial^{(D)}}
\newcommand{\pG}{\partial^{(G)}}

\newcommand{\ensemble}{ensemble }

\begin{document}
\maketitle
\begin{abstract}
% Neural networks enjoy a good deal of success across the whole gamut of applications in machine learning, but the theoretical underpinnings of their success remain limited. One particular theoretical question concerns the remarkable success of gradient descent based optimisation when applied to very large networks and extremely difficult learning problems on complex data. Several works in recent years have fruitfully studied simplified models of neural network loss surfaces, their geometry and the properties of gradient descent thereon. Modern machine learning applications use neural networks with many extra design features, beyond the basic multi-layer perceptron design that is typically considered theoretically. 
We present a novel mathematical model that seeks to capture the key design feature of generative adversarial networks (GANs). Our model consists of two interacting spin glasses, and we conduct an extensive theoretical analysis of the complexity of the model's critical points using techniques from Random Matrix Theory. The result is insights into the loss surfaces of large GANs that build upon prior insights for simpler networks, but also reveal new structure unique to this setting.

\end{abstract}

\section{Introduction}
By making various modeling assumptions about standard multi-layer perceptron neural networks, \cite{choromanska2015loss} argued heuristically that the training loss surfaces of large networks could be modelled by a spherical multi-spin glass. Using theoretical results of \cite{auffinger2013random}, they were able to arrive at quantitative asymptotic characterisations, in particular the existence of a favourable `banded structure' of local-optima of the loss. There are clear and acknowledged deficiencies with their assumptions \cite{choromanska2015open} and recent observations have shown that the Hessians of real-world deep neural networks do not behave like random matrices from the Gaussian Orthogonal Ensemble (GOE) of Random Matrix Theory at the macroscopic scale \cite{papyan2018full, granziol2019towards, granziol2020beyond}, despite this being implied by the spin-glass model of \cite{choromanska2015loss}. Moreover, there have been questions raised about whether the mean asymptotic properties of loss surfaces for deep neural networks (or energy surfaces of glassy objects) are even relevant practically for gradient-based optimisation in sub-exponential time \cite{baity2019comparing, mannelli2019passed, folena2019rethinking}, though interpretation of experiments with deep neural networks remains difficult and the discussion about the true shape of their loss surfaces and the implications thereof is far from settled. Nevertheless, spin-glass models  present a tractable example of high-dimensional complex random functions that may well provide insights into aspects of deep learning. Rather than trying to improve or reduce the assumptions of \cite{choromanska2015loss}, various authors have recently opted to skip the direct derivation from a neural network to a statistical physics model, instead proposing simple models designed to capture aspects of training dynamics and studying those directly. Examples include: the modified spin glass model of \cite{PhysRevX.9.011003} with some explicitly added `signal'; the simple explicitly non-linear model of \cite{maillard2019landscape}; the spiked tensor `signal-in-noise' model of \cite{mannelli2019afraid}. In a slightly different direction, \cite{baskerville2020loss} removed one of the main assumptions from the \cite{choromanska2015loss} derivation, and in so doing arrived at a deformed spin-glass model. All of this recent activity sits in the context of earlier work connecting spin-glass objects with simple neural networks \cite{kanter1987associative, gardner1988space, engel2001statistical} and, more generally, with image reconstruction and other signal processing problems \cite{nishimori2001statistical}.

% This work in particular demonstrated that the exact $p$ spin-glass correspondence relied upon by \cite{choromanska2015loss} is not necessary, hinting that the bottom-line implications of a banded loss surface structure may be, in some sense, universal.\\

One area that has not been much explored in the line of the above-mentioned literature is the study of architectural variants. Modern deep learning contains a very large variety of different design choices in network architecture, such as convolutional networks for image and text data (among others) \cite{goodfellow2016deep, conneau-etal-2017-deep}, recurrent networks for sequence data \cite{hochreiter1997long} and self-attention transformer networks for natural language \cite{devlin-etal-2019-bert, radford2018improving}. Given the ubiquity of convolutional networks, one might seek to study those, presumably requiring consideration of local correlations in data. One could imagine some study of architectural quirks such as residual connections \cite{he2016deep}, and batch-norm has been considered to some extent by \cite{pennington2017nonlinear}. In this work, we propose a novel model for \emph{generative adversarial networks} (GANs) \cite{NIPS2014_5423} as two interacting spherical spin glasses. GANs have been the focus of intense research and development in recent years, with a large number of variants being proposed \cite{radford2015unsupervised, zhang2018self, liu2016coupled,karras2020a, mirza2014conditional, arjovsky2017wasserstein, zhu2017unpaired} and rapid progress particularly in the field of image generation. From the perspective of optimisation, GANs have much in common with other deep neural networks, being complicated high-dimensional functions optimised using local gradient-based methods such as stochastic gradient descent and variants. On the other hand, the adversarial training objective of GANs, with two deep networks competing, is clearly an important distinguishing feature, and GANs are known to be more challenging to train than single deep networks. Our objective is to capture the essential adversarial aspect of GANs in a tractable model of high-dimensional random complexity which, though being a significant simplification, has established connections to neural networks and high dimensional statistics.

Our model is inspired by \cite{choromanska2015loss, ros2019complex, mannelli2019afraid, arous2019landscape} with spherical multi-spin glasses being used in place of deep neural networks. We thus provide a complicated, random, high-dimensional model with the essential feature of GANs clearly reflected in its construction. By employing standard Kac-Rice complexity calculations \cite{fyodorov2004complexity, fyodorov2007replica, auffinger2013random} we are able to reduce the loss landscape complexity calculation to a random matrix theoretic calculation. We then employ various Random Matrix Theory techniques as in \cite{baskerville2020loss} to obtain rigorous, explicit leading order asymptotic results. Our calculations rely on the supersymmetric method in Random Matrix Theory, in particular the approach to calculating limiting spectral densities follows \cite{Verbaarschot_2004} and the calculation also follows \cite{guhr1990isospin, guhr1991dyson} in important ways. The greater complexity of the random matrix spectra encountered present some challenges over previous such calculations, which we overcome with a combination of analytical and numerical approaches. Using our complexity results, we are able to draw qualitative implications about GAN loss surfaces analogous to those of \cite{choromanska2015loss} and also investigate the effect of a few key design parameters included in the GAN. We compare the effect of these parameters on our spin glass model and also on the results of experiments training real GANs. Our calculations include some novel details, in particular, we use precise sub-leading terms for a limiting spectral density obtained from supersymmetric methods to prove a required concentration result.

The role that statistical physics models such as spherical multi-spin glasses are to ultimately play in the theory of deep learning is not yet clear, with arguments both for and against their usefulness and applicability. We provide a first attempt to model an important architectural feature of modern deep neural networks within the framework of spin glass models and provide a detailed analysis of properties of the resulting loss (energy) surface. Our analysis reveals potential explanations for observed properties of GANs and demonstrates that it may be possible to inform practical hyperparameter choices using models such as ours. Much of the  advancement in practical deep learning has come from innovation in network architecture, so if deep learning theory based on simplified physics models like spin-glasses is to keep pace with practical advances in the field, then it will be necessary to account for architectural details within such models. Our work is a first step in that direction and the mathematical techniques used may prove more widely valuable.

The paper is structured as follows: in Section \ref{sec:model} we introduce the interacting spin glass model; in Section \ref{sec:kac} we use a Kac-Rice formula to derive random matrix expressions for the asymptotic complexity of our model; in Section \ref{sec:lsd} we derive the limiting spectral density of the relevant random matrix ensemble; in Section \ref{sec:complexity} we use the Coulomb gas approximation to compute the asymptotic complexity, and legitimise its use by proving a concentration result; in Section \ref{sec:implications} we derive some implications of our model for GAN training and compare to experimental results from real GANs; in Section \ref{sec:conclusions} we conclude. All code used for numerical calculations of our model, training real GANs, analysing the results and generating plots is made available\footnote{\url{https://github.com/npbaskerville/loss-surfaces-of-gans}}.

\section{An interacting spin glass model}\label{sec:model}
We use multi-spin glasses in high dimensions as a toy model for neural network loss surfaces without any further justification, beyond that found in \cite{choromanska2015loss, baskerville2020loss}. GANs are composed of two networks: \emph{generator} ($G$) and \emph{discriminator} ($D$). $G$ is a map $\mathbb{R}^m\rightarrow\mathbb{R}^d$ and $D$ is a map $\mathbb{R}^d\rightarrow\mathbb{R}$. $G$'s purpose is to generate synthetic data samples by transforming random input noise, while $D$'s is to distinguish between real data samples and those generated by $G$. Given some probability distribution $\mathbb{P}_{data}$ on some $\mathbb{R}^d$, GANs have the following minimax training objective \begin{align}\label{eq:gan}
    \min_{\Theta_G}\max_{\Theta_D}\left\{\mathbb{E}_{\vec{x}\sim \mathbb{P}_{data}} \log D(\vec{x}) + \mathbb{E}_{\vec{z}\sim \mathcal{N}(0, \sigma_z^2)}\log(1 - D(G(\vec{z})))\right\},
\end{align}
where $\Theta_D, \Theta_G$ are the parameters of the discriminator and generator respectively. With $\vec{z}\sim\mathcal{N}(0, \sigma_z^2)$, $G(\vec{z})$ has some probability distribution $\mathbb{P}_{gen}$. When successfully trained, the initially unstructured $\mathbb{P}_{gen}$ examples are easily distinguished by $D$, this in turn drives improvements in $G$, bring $\mathbb{P}_{gen}$ closer to $\mathbb{P}_{data}$. Ultimately, the process successfully terminates when $\mathbb{P}_{gen}$ is very close to $\mathbb{P}_{data}$ and $D$ performs little better than random at the distinguishing task. To construct our model, we introduce two spin glasses:
\begin{align}
    \lD(\vwD) &= \sum_{i_1,\ldots, i_p=1}^{N_D} X_{i_1,\ldots, i_p} \prod_{k=1}^p \wD_{i_k}\\
    \lG(\vwD, \vwG) &= \sum_{i_1,\ldots, i_p=1}^{N_D} \sum_{j_1,\ldots, j_q=1}^{N_G} Z_{i_1,\ldots, i_p, j_1,\ldots, j_q} \prod_{k=1}^p \wD_{i_k} \prod_{l=1}^q \wG_{i_l}\\  
\end{align}
where all the $X_{i_1,\ldots, i_p}$ are i.i.d. $\mathcal{N}(0,1)$ and $Z_{j_1,\ldots, j_q}$ are similarly i.i.d. $\mathcal{N}(0,1)$. We then define the generator and discriminator spin glasses: \begin{align}
    \LD(\vwD, \vwG) &= \lD(\vwD) - \sigma_z\lG(\vwD, \vwG),\label{eq:ld_def}\\
    \LG(\vwD, \vwG) &= \sigma_z \lG(\vwD,\vwG).\label{eq:lg_def}
\end{align}

$\lD$ plays the role of the loss of the discriminator network when trying to classify genuine examples as such. $\lG$ plays the role of loss of the discriminator when applied to samples produced by the generator, hence the sign difference between $\LD$ and $\LG$. $\vwD$ are the weights of the discriminator, and $\vwG$ the weights of the generator. The $X_{\vec{i}}$ are surrogates for the training data (i.e. samples from $\mathbb{P}_{data}$) and the $Z_{\vec{j}}$ are surrogates for the noise distribution of the generator. For convenience, we have chosen to pull the $\sigma_z$ scale outside of the $Z_{\vec{j}}$ and include it as a constant multiplier in (\ref{eq:ld_def})-(\ref{eq:lg_def}). In reality, we should like to keep $Z_{\vec{j}}$ as i.i.d. $\mathcal{N}(0,1)$ but take $X_{\vec{i}}$ to have some other more interesting distribution, e.g. normally or uniformly distributed on some manifold. Using $[x]$ to denote the integer part of $x$, we take $N_D = [\kappa N], N_G = [\kappa' N]$ for fixed $\kappa\in(0,1)$, $\kappa'=1-\kappa$, and study the regime $N\rightarrow\infty$. Note that there is no need to distinguish between $[\kappa N]$ and $\kappa N$ in the $N\rightarrow\infty$ limit.

\section{Kac-Rice formulae for complexity}\label{sec:kac}
Training GANs involves jointly minimising the losses of the discriminator and the generator. Therefore, rather than being interested simply in upper-bounding a single spin-glass and counting its stationary points, the complexity of interest comes from jointly upper bounding both $L^{(D)}$ and $L^{(G)}$ and counting points where both are stationary. Using $S^{M}$ to denote the $M$-sphere\footnote{We use the convention of the $M$-sphere being the sphere embedded in $\mathbb{R}^M$.}, we define the complexity

\begin{align}
    C_{N} =\Bigg|\left\{ \vwD\in S^{N_D  }, \vwG\in S^{N_G } ~:~ \nabla_D \LD = 0, \nabla_G \LG = 0, \LD\in B_D, \LG\in B_G\right\}\Bigg|
\end{align}
for some Borel sets $B_D, B_G\subset\mathbb{R}$ and where $\nabla_D, \nabla_G$ denote the conformal derivatives with respect to the discriminator and generator weights respectively.
Note:
\begin{enumerate}
    \item We have chosen to treat the parameters of each network as somewhat separate by placing them on their own hyper-spheres. This reflects the minimax nature of GAN training, where there really are 2 networks being optimised in an adversarial manner rather than one network with some peculiar structure.
    \item We could have taken $\nabla = (\nabla_D, \nabla_G)$ and required $\nabla \LD = \nabla \LG = 0$ but, as in the previous comment, our choice is more in keeping with the adversarial set-up, with each network seeking to optimize separately its own parameters in spite of the other.
    \item We will only be interested in the case $ B_D = (-\infty, \sqrtsign{N} u_D)$ and $B_G= (-\infty, \sqrtsign{N} u_G)$, for $u_D, u_G\in \mathbb{R}$.
\end{enumerate}

So that the finer structure of local minima and saddle points can be probed, we also define the corresponding complexity with Hessian index prescription
\begin{align}
    C_{N, k_D, k_G} =\Bigg|\Bigg\{ \vwD\in S^{N_D  }, \vwG\in S^{N_G } ~:~ &\nabla_D \LD = 0, \nabla_G \LG = 0, \LD\in B_D, \LG\in B_G\notag\\
    &i(\nabla_D^2 L^{(D)}) = k_D, ~i(\nabla_G^2 L^{(G)}) = k_G \Bigg\}\Bigg|,\label{eq:C_Nkk_def}
\end{align}
where $i(M)$ is the index of $M$ (i.e. the number of negative eigenvalues of $M$). To calculate the complexities, we follow the well-trodden route of Kac-Rice formulae as pioneered by \cite{fyodorov2004complexity, fyodorov2007replica}. For a fully rigorous treatment, we proceed as in \cite{auffinger2013random, baskerville2020loss} by turning to the following result from \cite{adler2009random}.

\begin{theorem}[\cite{adler2009random} Theorem 12.1.1]\label{thm:adler_kac_rice}
Let $\mathcal{M}$ be a compact , oriented, N-dimensional $C^1$ manifold with a $C^1$ Riemannian metric $g$. Let $\phi:\mathcal{M}\rightarrow\mathbb{R}^N$ and $\psi:\mathcal{M}\rightarrow \mathbb{R}^K$ be random fields on $\mathcal{M}$. For an open set $A\subset\mathbb{R}^K$ for which $\partial A$ has dimension $K-1$ and a point $\vec{u}\in\mathbb{R}^{N}$ let \begin{equation}
    N_{\vec{u}} \equiv \left|\{x\in\mathcal{M} ~|~ \phi(x) = \vec{u}, ~ \psi(x)\in A\}\right|.
\end{equation}

Assume that the following conditions are satisfied for some orthonormal frame field E:
\begin{enumerate}[label=(\alph*)]
\item
All components of $\phi$, $\nabla_E \phi$, and $\psi$ are a.s. continuous and have finite variances (over $\mathcal{M}$).
\item
 For all $x\in\mathcal{M}$, the marginal densities $p_{x}$  of $\phi(x)$ (implicitly assumed to exist) are continuous at  $\vec{u}$.
 \item
 The conditional densities $p_{x}(\cdot|\nabla_E\phi(x),\psi(x))$ of $\phi(x)$ given $\psi(x)$ and $\nabla_E\phi(x)$ (implicitly assumed to exist) are bounded above and continuous at $\vec{u}$, uniformly in $\mathcal{M}$.
 \item
 The conditional densities $p_x (\cdot|\phi(x) = \vec{z})$ of $\det(\nabla_{E_j}\phi^i  (x))$ given are continuous in a neighbourhood of $0$ for $\vec{z}$ in a neighbourhood of $\vec{u}$  uniformly in $\mathcal{M}$.
 \item
 The conditional densities $p_x (\cdot|\phi (x) = \vec{z})$ are continuous for $\vec{z}$ in a neighbourhood of $\vec{u}$ uniformly in $\mathcal{M}$.
 \item
 The following moment condition holds \begin{equation}
     \sup_{x\in\mathcal{M}}\max_{1\leq i,j\leq N}\expect\left\{\left|\nabla_{E_j}\phi^i(x)\right|^N\right\}< \infty
 \end{equation}
 \item
 The moduli of continuity with respect to the (canonical) metric induced  by $g$ of each component of $\psi$, each component of $\phi$ and each $\nabla_{E_j}\phi^i$ all satisfy, for any $\epsilon > 0$ \begin{equation}\label{eq:moduli_condition}
    \mathbb{P}( \omega(\eta) >\epsilon) = o(\eta^N), ~~ \text{as } \eta\downarrow 0
 \end{equation}
 where the \emph{modulus of continuity} of a real-valued function $G$ on a metric space $(T, \tau)$ is defined as (c.f. \cite{adler2009random} around (1.3.6)) \begin{equation}
     \omega(\eta) \coloneqq \sup_{s,t : \tau(s,t)\leq\eta}\left|G(s) - G(t)\right|
 \end{equation} 
\end{enumerate}
Then \begin{equation}\label{eq:adler_taylor_kac_rice}
    \expect N_{\vec{u}} = \int_{\mathcal{M}}\expect \left\{|\det \nabla_E\phi(x)|\indic\{\psi(x)\in A\} ~| ~ \phi(x) = \vec{u}\right\}p_x(\vec{u}) \text{Vol}_g(x)
\end{equation}
where $p_x$ is the density of $\phi$ and $\text{Vol}_g$ is the volume element induced by $g$ on $\mathcal{M}$.
\end{theorem}

In the notation of Theorem \ref{thm:adler_kac_rice}, we make the following choices: \begin{align*}
    \phi = \left(\begin{array}{c}
         \nabla_D \LD \\
         \nabla_G \LG
    \end{array}\right), ~~~  \psi = \left(\begin{array}{c}
          \LD \\
         \LG
    \end{array}\right)
\end{align*}
and so \begin{align*}
    A = B_D\times B_G, ~~~ \vec{u} = 0.
\end{align*}
and the manifold $\mathcal{M}$ is taken to be $S^{N_D}\times S^{N_G}$ with the product topology.\\

\begin{lemma}
\begin{align}
    C_N = \int_{S^{N_D}\times S^{N_G}}d\vwG d\vwD ~~\varphi_{(\nabla_D \LD,
         \nabla_G \LG)}(0)\expect&\left[  |\det \left(\begin{array}{cc} \nabla_D^2 \LD & \nabla_{GD} \LD \\ \nabla_{DG}\LG & \nabla^2_{G} \LG\end{array}\right)|  ~\Bigg|~ \nabla_G\LG =0, \nabla_D\LD = 0\right]\notag\\
         &\mathbbm{1}\left\{\LD\in B_D, \LG\in B_G\right\}\label{eq:kac_rice_first_old}
\end{align}
and therefore
\begin{align}
    C_N = \int_{S^{N_D}\times S^{N_G}}&d\vwG d\vwD ~~\varphi_{(\nabla_D \LD,
         \nabla_G \LG)}(0)\int_{B_D}dx_D \int_{B_G} dx_G ~ \varphi_{\LD}(x_D)\varphi_{\LG}(x_G) \notag\\ &\expect\Bigg[  |\det \left(\begin{array}{cc} \nabla_D^2 \LD & \nabla_{GD} \LD \\ \nabla_{DG}\LG  & \nabla^2_{G} \LG\end{array}\right)|  ~\Bigg|~ \nabla_G\LG =0, \nabla_D\LD = 0, 
         \LD = x_D, \LG = x_G\Bigg].\label{eq:kac_rice_first}
\end{align}

where $\varphi_{(\nabla_D \LD,\nabla_G \LG)}$ is the joint density of $(\nabla_D \LD,\nabla_G \LG)^T$, $\varphi_{\LD}$ the density of $\LD$, and $\varphi_{\LG}$ the density of $\LG$, all implicitly evaluated at $(\vwG, \vwD)$.
\end{lemma}
\begin{proof}
It is sufficient to check the conditions of Theorem \ref{thm:adler_kac_rice} with the above choices.

Conditions (a)-(f) are satisfied due to Gaussianity and the manifestly smooth definition of $L^{(D)}, L^{(G)}$. The moduli of continuity conditions as in (g) are satisfied separately for $L^{(D)}$ and its derivatives on $S^{N_D}$ and  for $L^{(G)}$ and its derivatives on $S^{N_G}$, as seen in the proof of the analogous result for a single spin glass in \cite{auffinger2013random}. But since $\mathcal{M}$ is just a direct product with product topology, it immediately follows that (g) is satisfied, so Theorem \ref{eq:adler_taylor_kac_rice} applies and we obtain (\ref{eq:kac_rice_first_old}). (\ref{eq:kac_rice_first}) follows simply, using the rules of conditional expectation.
\end{proof}

Define the Hessian matrix \begin{align*}
    \tilde{H} = \left(\begin{array}{cc} \nabla_D^2 \LD & \nabla_{GD} \LD \\ \nabla_{DG}\LG & \nabla^2_{G} \LG\end{array}\right).
\end{align*}
To make use of (\ref{eq:kac_rice_first}), we need the joint distribution of $\left(\lD, \pD_i \lD, \pD_{jk}\lD\right)$ and the independent  $\left(\lG, \pG_i \lG, \pG_{jk}\lG, \pD_l \lG, \pD_{mn}\lG \right)$. As in \cite{auffinger2013random}, we will simplify the calculation by evaluating in the region of the north poles on each hyper-sphere. $\lD$ behaves just like a single spin glass, and so we have \cite{auffinger2013random}:\begin{align}
    Var(\lD) &= 1,\label{eq:ld_var}\\
    Cov(\pD_i \lD, \pD_{jk} \lD) &= 0,\\
    \pD_{ij}\lD ~|~ \{\lD=x_D\} &\sim \sqrtsign{(N_D-1)p(p-1)}GOE^{N_D - 1} - x_DpI.
\end{align}
To find the joint and thence conditional distributions for $\lG$, we first note that \begin{align}
    Cov(\lG(\vwD, \vwG), \lG({\vwD}', {\vwG}')) = \left(\vwD\cdot {\vwD}' + \vwG \cdot {\vwG}'\right)^{p+q}
\end{align}
from which, by comparing with \cite{auffinger2013random}, one can write down the necessary expressions, at the north poles in a coordinate basis:
\begin{align}
    Var(\lG) &= 2^{p+q},\label{eq:var_lg}\\
    Cov(\pG_{ij}\lG, \lG) &= - (p+q)2^{p+q}\delta_{ij},\\
    Cov(\pD_{ij} \lG, \lG) &= - (p+q)2^{p+q}\delta_{ij},\\
    Cov(\pG_{ij}\lG,  \pG_{kl}\lG) &=  2^{p+q}\left[ (p+q)(p+q-1)\left(\delta_{ik}\delta_{jl} + \delta_{il}\delta_{jk}\right) + (p+q)^2 \delta_{ij}\delta_{kl}\right],\\
     Cov(\pG_{ij}\lG,  \pD_{kl}\lG) &=  2^{p+q} (p+q)^2 \delta_{ij}\delta_{kl},\\
      Cov(\pG_i\pD_j\lG,  \pG_{k}\pD_l\lG) &= 2^{p+q} (p+q)(p+q-1) \delta_{ik}\delta_{jl},\label{eq:ginibre}\\
      Cov(\pG_{ij}\lG,  \pG_{k}\pD_l\lG) &= 0\label{eq:gin_goe1}\\
    Cov(\pD_{ij}\lG,  \pD_{k}\pG_l\lG) &= 0\label{eq:gin_goe2}.
\end{align}
Also, all first derivatives of $\lG$ are clearly independent of $\lG$ and its second derivatives by the same reasoning as in \cite{auffinger2013random}. Note that \begin{align}
    Cov(\partial^{(D)}_i L^{(D)}, \partial^{(D)}_j L^{(D)})& = (p + \sigma_z^2 2^{p+q}(p+q))\delta_{ij}\\
    Cov(\partial^{(G)}_iL^{(G)}, \partial^{(G)}_j L^{(G)})& = \sigma^2_z 2^{p+q}(p+q)\delta_{ij}\\
     Cov(\partial^{(D)}_iL^{(D)}, \partial^{(G)}_j L^{(G)})& = 0
\end{align}
and so \begin{align}\label{eq:grad_dens_0}
    \varphi_{\left(\nabla_D L^{(D)}, \nabla_G L^{(G)}\right)}(0) = (2\pi)^{-\frac{N-2}{2}} \left(p + \sigma_z^22^{p+1}(p+q)\right)^{-\frac{N_D - 1}{2}} \left(\sigma_z^2 2^{p+q} (p+q)\right)^{-\frac{N_G-1}{2}}.
\end{align}
We need now to calculate the joint distribution of $(\pD_{ij}\lG, \pG_{kl}\lG)$ conditional on $\{\lG = x_G\}$. Denote the covariance matrix for $(\pD_{ij}\lG, \pG_{kl}\lG, \lG)$ by \begin{align}
    \Sigma = \left(\begin{array}{cc}
         \Sigma_{11}&\Sigma_{12}  \\
         \Sigma_{21}&\Sigma_{22} 
    \end{array}\right)
\end{align}
where \begin{align}
    \Sigma_{11} &= 2^{p+q}\left(\begin{array}{cc}
       (p+1)(p+q-1)(1 + \delta_{ij}) + (p+q)^2\delta_{ij}  & (p+q)^2 \delta_{ij}\delta_{kl} \\
             (p+q)^2 \delta_{ij}\delta_{kl} & (p+1)(p+q-1)(1 + \delta_{kl}) + (p+q)^2\delta_{kl}  
    \end{array}\right),\\
    \Sigma_{12} &= -2^{p+q} (p+q)\left(\begin{array}{c}
         \delta_{ij}  \\
          \delta_{kl}
    \end{array}\right),\\
      \Sigma_{21} &= -2^{p+q} (p+q)\left(\begin{array}{cc}
         \delta_{ij}  &\delta_{kl}
    \end{array}\right),\\
    \Sigma_{22} &= 2^{p+q}.
\end{align}
The conditional covariance is then \begin{align}
    \bar{\Sigma} = 2^{p+q}(p+1)(p+q-1)\left(\begin{array}{cc}
       1 + \delta_{ij} & 0 \\
            0 & 1 + \delta_{kl}
    \end{array}\right).\label{eq:double_goe}
\end{align}

In summary, from (\ref{eq:double_goe}) and (\ref{eq:ginibre}-\ref{eq:gin_goe2}) we obtain \begin{align}
    \left(\begin{array}{cc}
        -\nabla_D^2\lG & -\nabla_G\nabla_D \lG  \\
         \nabla_D\nabla_G \lG & \nabla^2 \lG  
    \end{array}\right) ~|~ \{\lG = x_G\} &\overset{d}{=} 
   \sqrtsign{2^{p+q+1}(p+q)(p+q-1)} \left(\begin{array}{cc}
       \sqrtsign{N_D -1}M^{(D)}_1 & -2^{-1/2}G \\
         2^{-1/2}G^T & \sqrtsign{N_G - 1}M^{(G)}
    \end{array}\right)\notag\\
    &~~~~~~- (p+q)x_G2^{p+1} \left(\begin{array}{cc}
        -I_{N_D} & 0  \\
         0 & I_{N_G} 
    \end{array}\right)
\end{align}
where $M^{(D)}_1\sim GOE^{N_D - 1}$ and $M^{(G)} \sim GOE^{N_G - 1}$ are independent GOEs and $G$ is an independent $N_D - 1 \times N_G - 1$ Ginibre matrix with entries of unit variance. Therefore, conditional on $\{(\lD, \lG) = (x_D, x_G)\}$, \begin{align}
    \tilde{H} \overset{d}{=}
  \sqrtsign{2p(p-1)} \left(\begin{array}{cc}
       \sqrtsign{N_D -1}M^{(D)}_2 & 0 \\
         0 & 0
    \end{array}\right) +&  \sigma_z\sqrtsign{2^{p+q+1}(p+q)(p+q-1)} \left(\begin{array}{cc}
       \sqrtsign{N_D -1}M^{(D)}_1 & -2^{-1/2}G \\
        2^{-1/2} G^T & \sqrtsign{N_G - 1}M^{(G)}
    \end{array}\right) \notag\\
   -& \sigma_z(p+q)x_G2^{p+q} \left(\begin{array}{cc}
        -I_{N_D} & 0  \\
         0 & I_{N_G} 
    \end{array}\right)  - px_D\left(\begin{array}{cc}
        I_{N_D} & 0  \\
         0 & 0
    \end{array}\right)  
\end{align}
where $M^{(D)}_2$ is another independent $GOE^{N_D - 1}$ matrix. We can simplify to obtain:
\begin{align}
    \tilde{H} =  \left(\begin{array}{cc}
       \sigma_D\sqrtsign{N_D -1}M^{(D)} & -2^{-1/2}\sigma_GG, \\
       2^{-1/2} \sigma_G G^T & \sigma_G\sqrtsign{N_G - 1}M^{(G)}
    \end{array}\right)
   -\sigma_z(p+q)x_G2^{p+q} \left(\begin{array}{cc}
        -I_{N_D} & 0  \\
         0 & I_{N_G} 
    \end{array}\right)  - px_D\left(\begin{array}{cc}
        I_{N_D} & 0  \\
         0 & 0
    \end{array}\right) \label{eq:H_first}
\end{align}
where \begin{align}\sigma_G&= \sigma_z\sqrtsign{2^{p+q+1}(p+q)(p+q-1)}\\
\sigma_D &= \sqrtsign{\sigma_G^2 + 2p(p-1)}
\end{align}
and $M^{(D)}\sim GOE^{N_D - 1}$ is a GOE matrix independent of $M^{(G)}$ and $G$.\\

Alternatively, because $M^{(D)}_{1,2} \overset{d}{=} -M^{(D)}_{1,2}$, let us write $\tilde{H}$ as \begin{align}
    \tilde{H} =& \sigma_zJ\left(\sqrtsign{2^{p+q+1}(p+q)(p+q-1)(N_D + N_G - 2)}M_1 - (p+q)x_G2^{p+q}I\right) \notag\\
  &+  \left(\sqrtsign{2p(p-1)(N_D - 1)}\left(\begin{array}{cc} M_2 & 0 \\ 0 & 0 \end{array}\right) - px_D\left(\begin{array}{cc} I_{N_D} & 0 \\ 0 & 0 \end{array}\right)\right)\notag\\
  \overset{d}{=}& J\Bigg[ \sigma_z\sqrtsign{2^{p+q+1}(p+q)(p+q-1)(N_D + N_G - 2)}M_1 - \sigma_z(p+q)x_G2^{p+q}I\notag\\
  &+ \sqrtsign{2p(p-1)(N_D - 1)}\left(\begin{array}{cc} M_2 & 0 \\ 0 & 0 \end{array}\right) + px_D\left(\begin{array}{cc} I_{N_D} & 0 \\ 0 & 0 \end{array}\right)\Bigg]
\end{align}

where $M_1\sim GOE^{N_D + N_G - 2}$ is a GOE matrix of size $N_D + N_G-2$, $M_2\sim GOE^{N_D - 1}$ is a GOE matrix of size $N_D-1$ and \begin{align}
    J = \left(\begin{array}{cc} -I_{N_D} & 0 \\ 0 & I_{N_G} \end{array}\right).
\end{align}
If follows that \begin{align}
    |\det \tilde{H}| \overset{d}{=}& \Bigg|\det\Bigg[ \sigma_z\sqrtsign{2^{p+q+1}(p+q)(p+q-1)(N_D + N_G - 2)}M_1 - \sigma_z(p+q)x_G2^{p+q}I\notag\\
  &+ \sqrtsign{2p(p-1)(N_D - 1)}\left(\begin{array}{cc} M_2 & 0 \\ 0 & 0 \end{array}\right) + px_D\left(\begin{array}{cc} I_{N_D} & 0 \\ 0 & 0 \end{array}\right)\Bigg]\Bigg|.\label{eq:H_rewritten}
\end{align}

Let $b = \sqrtsign{2^{p+q}(p+q)(p+q-1)}\sigma_z$, $b_1= \sqrtsign{p(p-1)\kappa}$ and \begin{align}
    x = \frac{\sigma_z(p+q)2^{p+q}}{\sqrtsign{N-2}}x_G, ~~~~  x_1= -\frac{p}{\sqrtsign{N-2}}x_D.\label{eq:x_x1_def}
\end{align}  Then \begin{align}
    |\det \tilde{H}| &= \left(2(N-2)\right)^{\frac{N-2}{2}} |\det H(x, x_1)|,\\
    H(x, x_1) &\equiv bM + b_1\left(\begin{array}{cc} M_1 & 0 \\ 0 & 0 \end{array}\right) - x - x_1\left(\begin{array}{cc} I_{N_D} & 0 \\ 0 & 0 \end{array}\right).
\end{align}

The desired complexity term then comes from (\ref{eq:kac_rice_first}), (\ref{eq:grad_dens_0}) \begin{align}
    \mathbb{E} C_N = K_N \int_B    \sqrtsign{\frac{N-2}{2\pi s^2}}e^{-\frac{N-2}{2s^2}x^2}dx ~ \sqrtsign{\frac{N-2}{2\pi s_1^2}} e^{ -\frac{N-2}{2s_1^2} x_1^2}dx_1 \mathbb{E}|\det H(x, x_1)|\label{eq:ecn}
\end{align}
where \begin{align}
    K_N =\omega_{\kappa N}\omega_{\kappa'N} (2(N-2))^{\frac{N-2}{2}} (2\pi)^{-\frac{N-2}{2}}\left(p + \sigma_z^22^{p+1}(p+q)\right)^{-\frac{\kappa N-1}{2}} \left(\sigma_z^2 2^{p+q} (p+q)\right)^{-\frac{\kappa' N-1}{2}}\label{eq:KN_def}
\end{align}
 and \begin{align}
    B =  \left\{(x, x_1)\in\mathbb{R}^2 ~:~ x\leq \frac{1}{\sqrtsign{2}}(p+q)2^{p+q} u_G, ~~ x_1 \geq -(p+q)^{-1} 2^{-(p+q)}p x - \frac{p}{\sqrtsign{2}}u_D\right\}
 \end{align}
 and the variances are (recall (\ref{eq:ld_var}), (\ref{eq:var_lg}), (\ref{eq:x_x1_def})) \begin{align}
    s^2 = \frac{1}{2}\sigma_z^2(p+q)^2 2^{3(p+q)}, ~~~ s_1^2 = \frac{p^2}{2},
\end{align}
 and $\omega_N = \frac{2\pi^{N/2}}{\Gamma(N/2)}$ is the surface area of the $N$ sphere. The domain of integration $B$ arises from the contraints $ L^{(D)} \in (-\infty, \sqrtsign{N} u_D)$ and $L^{(G)} \in (-\infty, \sqrtsign{N} u_G)$. 
 
 We will need the asymptotic behaviour of the constant $K_N$, which we now record in a small lemma.

\begin{lemma}\label{lemma:kn}
As $N\rightarrow\infty$, \begin{align}
K_N \sim 2^{\frac{N}{2}}\pi^{N/2} \left(\kappa^{\kappa} \kappa'^{\kappa'}\right)^{-N/2} \sqrtsign{\kappa\kappa'}\left(p + \sigma_z^22^{p+1}(p+q)\right)^{-\frac{\kappa N-1}{2}} \left(\sigma_z^2 2^{p+q} (p+q)\right)^{-\frac{\kappa' N-1}{2}}
\end{align}
\end{lemma}

\begin{proof} 
By Stirling's formula \begin{align}
    K_N &\sim 4 \pi^{N} \left(\frac{4\pi}{\kappa N}\right)^{-1/2}\left(\frac{4\pi}{\kappa' N}\right)^{-1/2}\left(\frac{\kappa N}{2 e}\right)^{-\kappa N/2}
    \left(\frac{\kappa' N}{2 e}\right)^{-\kappa' N/2}
    \left(2(N-2)\right)^{\frac{N-2}{2}} \left(2\pi\right)^{-\frac{N-2}{2}}\notag\\
    &~~~~~~ \left(p + \sigma_z^22^{p+1}(p+q)\right)^{-\frac{\kappa N-1}{2}} \left(\sigma_z^2 2^{p+q} (p+q)\right)^{-\frac{\kappa' N-1}{2}}\notag\\
    & \sim 2^{\frac{N}{2}}\pi^{N/2} \left(\kappa^{\kappa} \kappa'^{\kappa'}\right)^{-N/2} \sqrtsign{\kappa\kappa'}\left(p + \sigma_z^22^{p+1}(p+q)\right)^{-\frac{\kappa N-1}{2}} \left(\sigma_z^2 2^{p+q} (p+q)\right)^{-\frac{\kappa' N-1}{2}}
\end{align}
where we have used $(N-2)^{\frac{N-2}{2}} = N^{\frac{N-2}{2}} \left(1- \frac{2}{N}\right)^{\frac{N-2}{2}} \sim N^{\frac{N-2}{2}} e^{-N/2}.$
\end{proof}

 \section{Limiting spectral density of the Hessian}\label{sec:lsd}
Our intention now is to compute the the expected complexity $\expect C_N$ via the Coulomb gas method. The first step in this calculation is to obtain the limiting spectral density of the random matrix \begin{align}
    H' = bM + b_1\left(\begin{array}{cc}
     M_1 & 0 \\ 0 & 0\end{array}\right) - x_1\left(\begin{array}{cc}
     I & 0 \\ 0 & 0\end{array}\right),
\end{align}
where, note, $H = H' - xI$. Here the upper-left block is of dimension $\kappa N$, and the overall dimension is $N$. Let $\mu_{eq}$ be the limiting spectral measure of $H'$ and $\rho_{eq}$ its density. The supersymmetric method provides a way of calculating the expected Stieltjes transforms of $\rho_{eq}$ \cite{Verbaarschot_2004}:\begin{align}
    \langle G(z) \rangle &= \frac{1}{N} \frac{\partial}{\partial J}\Bigg|_{J=0} Z(J)\\
    Z(J) &\coloneqq \mathbb{E}_{H'} \frac{\det(z - H' + J)}{\det(z - H')}.
\end{align}
Recall that a density and its Stieltjes transform are related by the Stieltjes inversion formula  \begin{align}\label{eq:stieljes}
     \rho_{eq}(z) = \frac{1}{\pi}\lim_{\epsilon \rightarrow 0} \Im \langle G(z + i\epsilon)\rangle.
 \end{align}
The function $Z(J)$ can be computed using a supersymmetric representation of the ratio of determinants. Firstly, we recall an elementary result from multivariate calculus, where $M$ is a real matrix: \begin{align}\label{eq:det_complex_id}
    \int \prod_{i=1}^N \frac{d\phi_i d\phi_i^*}{2\pi} e ^{-i\phi^{\dagger}M\phi} = \frac{1}{\det M}.
\end{align}
By introducing the notion of \emph{Grassmann variables} and \emph{Berezin integration}, we obtain a complimentary expression: \begin{align}\label{eq:det_grass_id}
    \int \prod_{i=1}^N \frac{d\chi_i d\chi_i^*}{-i} e^{-i\chi^{\dagger}M\chi} = \det{M}.\end{align}
Here the $\chi_i, \chi_i^*$ are purely algebraic objects defined by the anti-commutation rule \begin{align}
    \chi_i\chi_j = - \chi_j\chi_i, ~~ \forall i,j\label{eq:anticom_def}
\end{align}
and $\chi_i^*$ are separate objects, with the complex conjugation unary operator ${}^*$ defined so that $\left(\chi_i^*\right)^* = -\chi_i^*,$ and Hermitian conjugation is then defined as usual by $\chi^{\dagger} = (\chi^T)^*.$ The set of variables $\{\chi_i, \chi_i^*\}_{i=1}^N$ generate a \emph{graded algebra} over $\mathbb{C}$. Mixed vectors of commuting and anti-commuting variables are called \emph{supervectors}, and they belong to a vector space called \emph{superspace}. The integration symbol $\int d\chi_id\chi^*$ is defined as a formal algebraic linear operator by the properties \begin{align}\label{eq:berezin_def}
    \int d\chi_i = 0, ~~~~~ \int d\chi_i ~\chi_j = \delta_{ij}.
\end{align}
Functions of the the Grassmann variables are defined by their formal power series, e.g. \begin{align}
    e^{\chi_i} = 1 + \chi_i + \frac{1}{2}\chi_i^2 + \ldots = 1 + \chi_i
\end{align}
where the termination of the series follows from $\chi_i^2 = 0 ~~ \forall i$, which is an immediate consequence of (\ref{eq:anticom_def}). From this it is apparent that (\ref{eq:berezin_def}), along with (\ref{eq:anticom_def}), is sufficient to define Berezin integration over arbitrary functions of arbitrary combinations of Grassmann variables.
Using the integral results (\ref{eq:det_complex_id}), (\ref{eq:det_grass_id})  we can then write \begin{align}
 \frac{\det(z - H' + J)}{\det(z - H')} &= \int d\Psi \exp\left\{ -i\phi^{\dagger}(z-H') \phi - i \chi^{\dagger}(z+J - H')\chi\right\}\label{eq:susy_lsd_ratio}
 \end{align}
 where the measure is \begin{align}
     d\Psi = \prod_{t=1}^2 \frac{d\phi[t]d\phi^{*}[t] d\chi[t] d\chi^*[t]}{-2\pi i},
 \end{align}
$\phi$ is a vector of $N$ complex commuting variables, $\chi$ and  $\chi^*$ are vectors of $N$ Grassmann variables, and we use the $[t]$ notation to denote the splitting of each of the vectors into the first $\kappa N$ and last $(1-\kappa)N$ components, as seen in \cite{guhr1990isospin}: \begin{align}
    \phi = \left(\begin{array}{c} \phi[1] \\ \phi[2]\end{array}\right).
\end{align}
We then split the quadratic form expressions in (\ref{eq:susy_lsd_ratio}) \begin{align}
    & -\phi^{\dagger}(z-H') \phi -  \chi^{\dagger}(z+J - H')\chi\notag\\
    = & -\phi[1]^{\dagger}(x_1 - b_1M_1) \phi[1]
     -\phi^{\dagger}(z - bM) \phi
    - \chi[1]^{\dagger}(x_1 - b_1 M_1)\chi[1]
        - \chi^{\dagger}(z + J - b M)\chi.
\end{align}
Taking the GOE averages is now simple \cite{Verbaarschot_2004,nock2017characteristic}:\begin{align}
    \mathbb{E}_M \exp\left\{-ib\phi^{\dagger}M \phi - ib\chi^{\dagger}M\chi\right\} &= \exp\left\{- \frac{b^2}{4N}\trg Q^2\right\},\\
  \mathbb{E}_M \exp\left\{-ib_1\phi[1]^{\dagger}M_1 \phi[1] - ib_1\chi[1]^{\dagger}M_1\chi[1]\right\} &= \exp\left\{- \frac{b_1^2}{4\kappa N}\trg Q[1]^2\right\},
\end{align}
where the supersymmetric matrices are given by \begin{align}
    Q = \left(\begin{array}{cc} \phi^{\dagger}\phi & \phi^{\dagger}\chi \\ \chi^{\dagger}\phi & \chi^{\dagger}\chi\end{array}\right), ~~~    Q[1] = \left(\begin{array}{cc} \phi[1]^{\dagger}\phi[1] & \phi[1]^{\dagger}\chi[1] \\ \chi[1]^{\dagger}\phi[1] & \chi[1]^{\dagger}\chi[1]\end{array}\right).
\end{align}
Introducing the tensor notation \begin{align}
    \psi = \phi \otimes \left(\begin{array}{c} 1 \\ 0 \end{array}\right) + \chi \otimes \left(\begin{array}{c} 0 \\ 1 \end{array}\right), ~~     \psi[1] = \phi[1] \otimes \left(\begin{array}{c} 1 \\ 0 \end{array}\right) + \chi[1] \otimes \left(\begin{array}{c} 0 \\ 1 \end{array}\right)
\end{align}
and \begin{align}
    \zeta = \left(\begin{array}{cc} z & 0 \\ 0 & z + J\end{array}\right)
\end{align}we can compactly write \begin{align}
    Z(J) = \int d\Psi \exp\left\{ - \frac{b^2}{4N} \trg Q^2 - \frac{b_1^2}{4\kappa N} \trg Q[1]^2 - i \psi[1]^{\dagger}\psi[1]x_1 - i \psi^{\dagger}\zeta \psi\right\}.
\end{align}
We now perform two Hubbard-Stratonovich transformations \cite{Verbaarschot_2004} \begin{align}
    Z(J) &= \int d\Psi d\sigma d\sigma[1]\exp\left\{ - \frac{N}{b^2} \trg \sigma^2 - \frac{\kappa N}{b_1^2} \trg \sigma[1]^2 - i \psi[1]^{\dagger}(x_1 + \sigma[1])\psi[1] - i \psi^{\dagger}(\sigma + \zeta )\psi\right\},
\end{align}
where $\sigma$ and $\sigma[1]$ inherit their form from $Q, Q[1]$ \begin{align}
    \sigma = \left( \begin{array}{cc} \sigma_{BB}& \sigma_{BF} \\ \sigma_{FB} & i\sigma_{FF}\end{array}\right), ~~~     \sigma[1] = \left( \begin{array}{cc} \sigma_{BB}[1] & \sigma_{BF}[1] \\ \sigma_{FB}[1] & i\sigma_{FF}[1]\end{array}\right)
\end{align}
with $\sigma_{BB}, \sigma_{FF}, \sigma_{BB}[1], \sigma_{FF}[1]$ real commuting variables, and $\sigma_{BF}, \sigma_{FB}, \sigma_{BF}[1], \sigma_{FB}[1]$ Grassmanns; the factor $i$ is introduced to ensure convergence. Integrating out over $d\Psi$ is now a straightforward Gaussian integral in superspace, giving \begin{align}
    Z(J) &= \int d\Psi d\sigma d\sigma[1] \exp\left\{ - \frac{N}{b^2} \trg \sigma^2 - \frac{\kappa N}{b_1^2} \trg \sigma[1]^2 - i \psi[1]^{\dagger}(x_1 + \zeta + \sigma + \sigma[1])\psi[1] - i \psi[2]^{\dagger}(\sigma + \zeta )\psi[2]\right\}\notag\\
    & = \int d\sigma d\sigma[1] \exp\left\{ - \frac{N}{b^2} \trg \sigma^2 - \frac{\kappa N}{b_1^2} \trg \sigma[1]^2 - \kappa N\trg\log(x_1 + \zeta + \sigma + \sigma[1]) -  \kappa' N \trg\log (\sigma + \zeta )\right\}\notag\\
    &=\int d\sigma d\sigma[1] \exp\left\{ - \frac{N}{b^2} \trg (\sigma - \zeta)^2 - \frac{\kappa N}{b_1^2} \trg \sigma[1]^2 - \kappa N\trg\log(x_1  + \sigma + \sigma[1]) -  \kappa' N \trg\log \sigma\right\}.
\end{align}
Recalling the definition of $\zeta$, we have \begin{align}
    \trg(\sigma - \zeta)^2 = (\sigma_{BB} - z)^2 - (i\sigma_{FF} - z- J)^2 
\end{align}
and so one immediately obtains \begin{align}
    \frac{1}{N}\frac{\partial }{\partial J}\Bigg|_{J=0} Z(J) &= \frac{2}{b^2}\int d\sigma d\sigma[1] (z - i\sigma_{FF}) \exp\left\{ - \frac{N}{b^2} \trg (\sigma - z)^2 - \frac{\kappa N}{b_1^2} \trg \sigma[1]^2 - \kappa N\trg\log(x_1  + \sigma + \sigma[1]) -  \kappa' N \trg\log \sigma\right\}\notag\\
    &= \frac{2}{b^2}\int d\sigma d\sigma[1] (z - i\sigma_{FF}) \exp\left\{ - \frac{N}{b^2} \trg \sigma ^2 - \frac{\kappa N}{b_1^2} \trg \sigma[1]^2 - \kappa N\trg\log(x_1  + z + \sigma + \sigma[1]) -  \kappa' N \trg\log ( z+\sigma)\right\}\label{eq:lsd_pre_saddle}
\end{align}
To obtain the limiting spectral density (LSD), or rather its Stieltjes transform, one must find the leading order term in the $N\rightarrow\infty$ expansion for (\ref{eq:lsd_pre_saddle}). This can be done by using the saddle point method on the $\sigma,\sigma[1]$ manifolds. We know that the contents of the exponential must vanish at the saddle point, since the LSD is $\mathcal{O}(1)$, so we in fact need only compute $\sigma_{FF}$ at the saddle point. We can diagonalise $\sigma$ within the integrand of (\ref{eq:lsd_pre_saddle}) and absorb the diagonalising graded $U(1/1)$ matrix into $\sigma[1]$. The resulting saddle point equations for the off-diagonal entries of the new (rotated) $\sigma[1]$ dummy variable are trivial and immediately give that $\sigma[1]$ is also diagonal at the saddle point. The saddle point equations are then \begin{align}
    &\frac{2}{b_1^2}\sigma_{BB}[1] + \frac{1}{\sigma_{BB}[1] + \sigma_{BB} + x_1 + z} = 0\label{eq:boson_saddle1}\\
        &\frac{2}{b^2}\sigma_{BB} + \frac{\kappa}{\sigma_{BB}[1] + \sigma_{BB} + x_1 + z} + \frac{\kappa'}{\sigma_{BB} + x} = 0\label{eq:boson_saddle2}\\
 &\frac{2}{b_1^2}\sigma_{FF}[1] - \frac{1}{\sigma_{FF}[1] + \sigma_{FF} - ix_1 - iz} = 0\label{eq:ferm_saddle1}\\
&\frac{2}{b^2}\sigma_{FF} - \frac{\kappa}{\sigma_{FF}[1] + \sigma_{FF} - ix_1 -iz} - \frac{\kappa'}{\sigma_{FF} - iz} = 0.\label{eq:ferm_saddle2}
\end{align}

(\ref{eq:ferm_saddle1}) and (\ref{eq:ferm_saddle2}) combine to give an explicit expression for $\sigma_{FF}[1]$: \begin{align}
    \sigma_{FF}[1] = \frac{b_1^2}{2\kappa}\left(\frac{2}{b^2}\sigma_{FF} - \kappa'(\sigma_{FF} - iz)^{-1}\right)\label{eq:ferm_saddle3}.
\end{align}
With a view to simplifying the numerical solution of the coming quartic, we define $t = i(\sigma_{FF} - iz)$ and then a line of manipulation with (\ref{eq:ferm_saddle2}) and (\ref{eq:ferm_saddle3}) gives \begin{align}
    \left(t^2 -  zt - \kappa' b^2\right)\left((1 + \kappa^{-1}b^{-2} b_1^2)t^2 - (\kappa^{-1}b_1^2 b^{-2}z - x_1)t - \kappa'\kappa^{-1}b_1^2\right) + b^2\kappa t^2= 0\label{eq:master_quartic}.
\end{align}
% The Stieltjes inversion formula (\ref{eq:stieljes}) and (\ref{eq:lsd_pre_saddle}) give the relationship between $t$ and the LSD: \begin{align}\label{eq:t_lsd}
%     \rho_{eq}(x) = -\frac{2}{b^2\pi} \Im t(x).
% \end{align}
 By solving (\ref{eq:master_quartic}) numerically for fixed values of $\kappa, b, b_1, x_1$, we can obtain the four solutions $t_1(z), t_2(z), t_3(z), t_4(z)$. These four solution functions arise from choices of branch for $(z, x_1)\in \mathbb{C}^2$ and determining the correct branch directly is highly non-trivial. However, for any $z\in\mathbb{R}$, at most one of the $t_i$ will lead to a positive LSD, which gives a simple way to compute $\rho_{eq}$ numerically using (\ref{eq:stieljes}) and (\ref{eq:lsd_pre_saddle}): \begin{align}\label{eq:t_lsd}
    \rho_{eq}(z) = \max_i\left\{-\frac{2}{b^2\pi} \Im t_i(z)\right\}.
\end{align}
Plots generated using (\ref{eq:t_lsd}) and eigendecompositions of matrices sampled from the distribution of $H'$ are given in Figure \ref{fig:spectra} and show good agreement between the two.
\begin{figure}[h]
    \centering
\begin{tabular}{ccc}
\subfloat[Merged]{\includegraphics[width=0.33\textwidth]{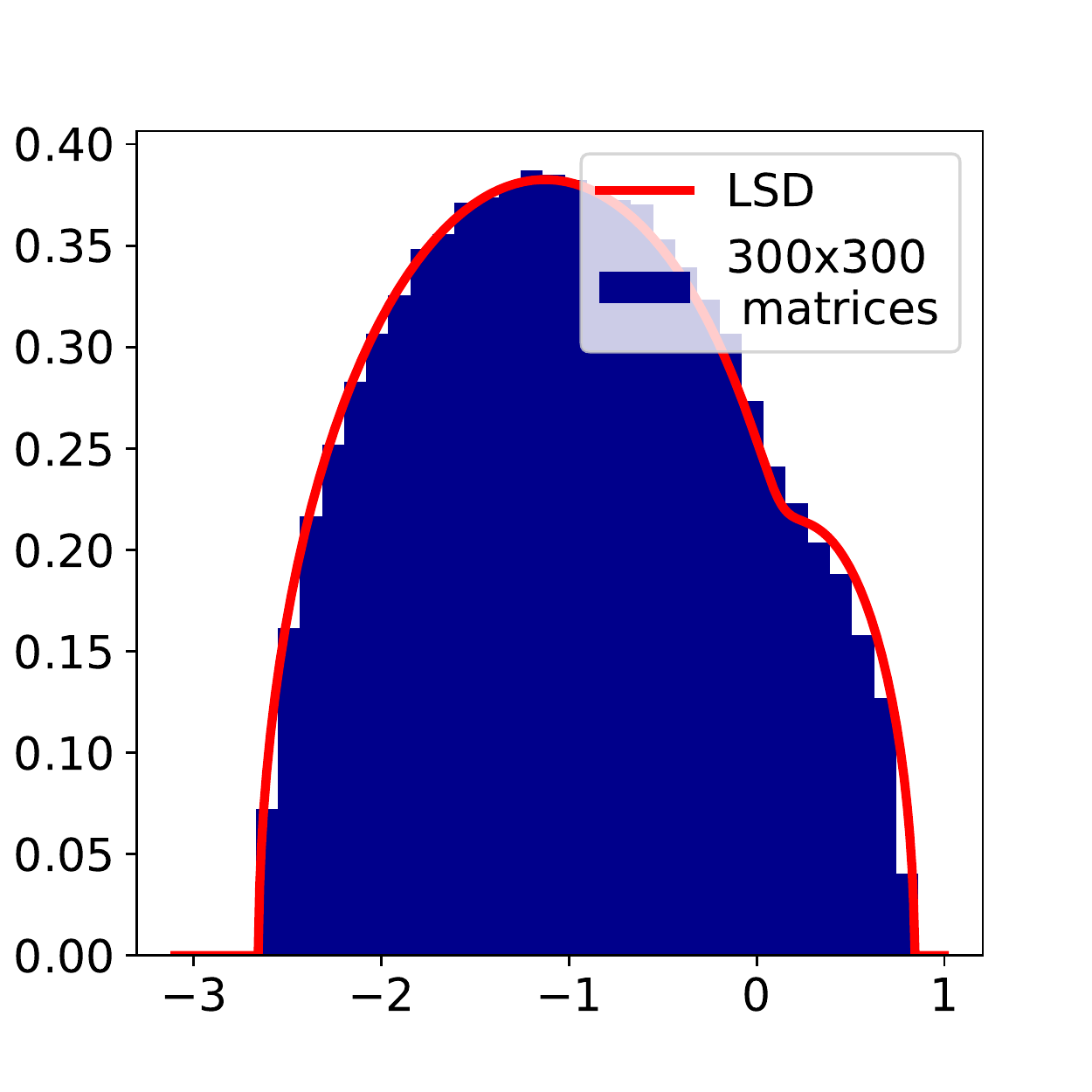}} &
\subfloat[Touching]{\includegraphics[width=0.33\textwidth]{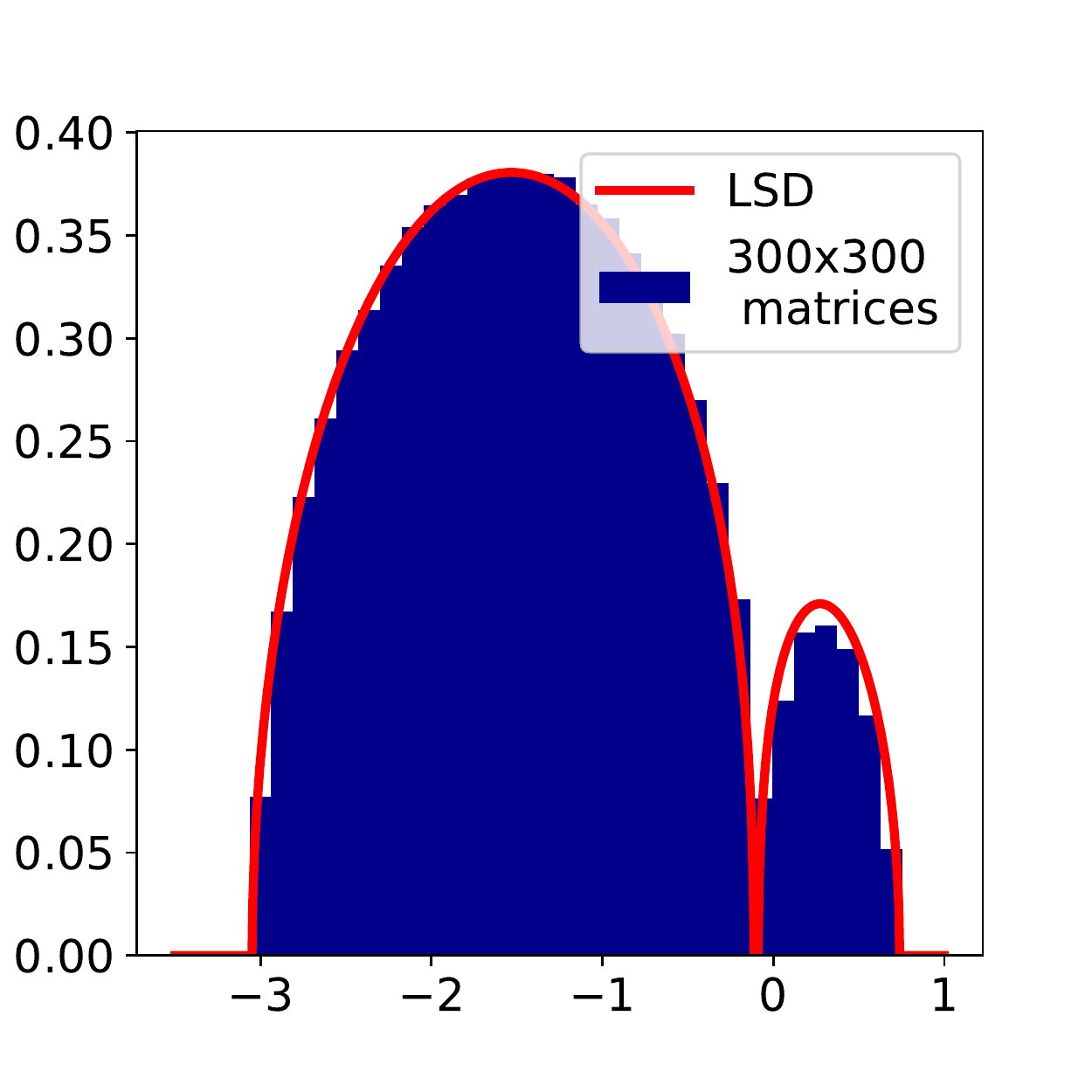}} &
\subfloat[Separate]{\includegraphics[width=0.33\textwidth]{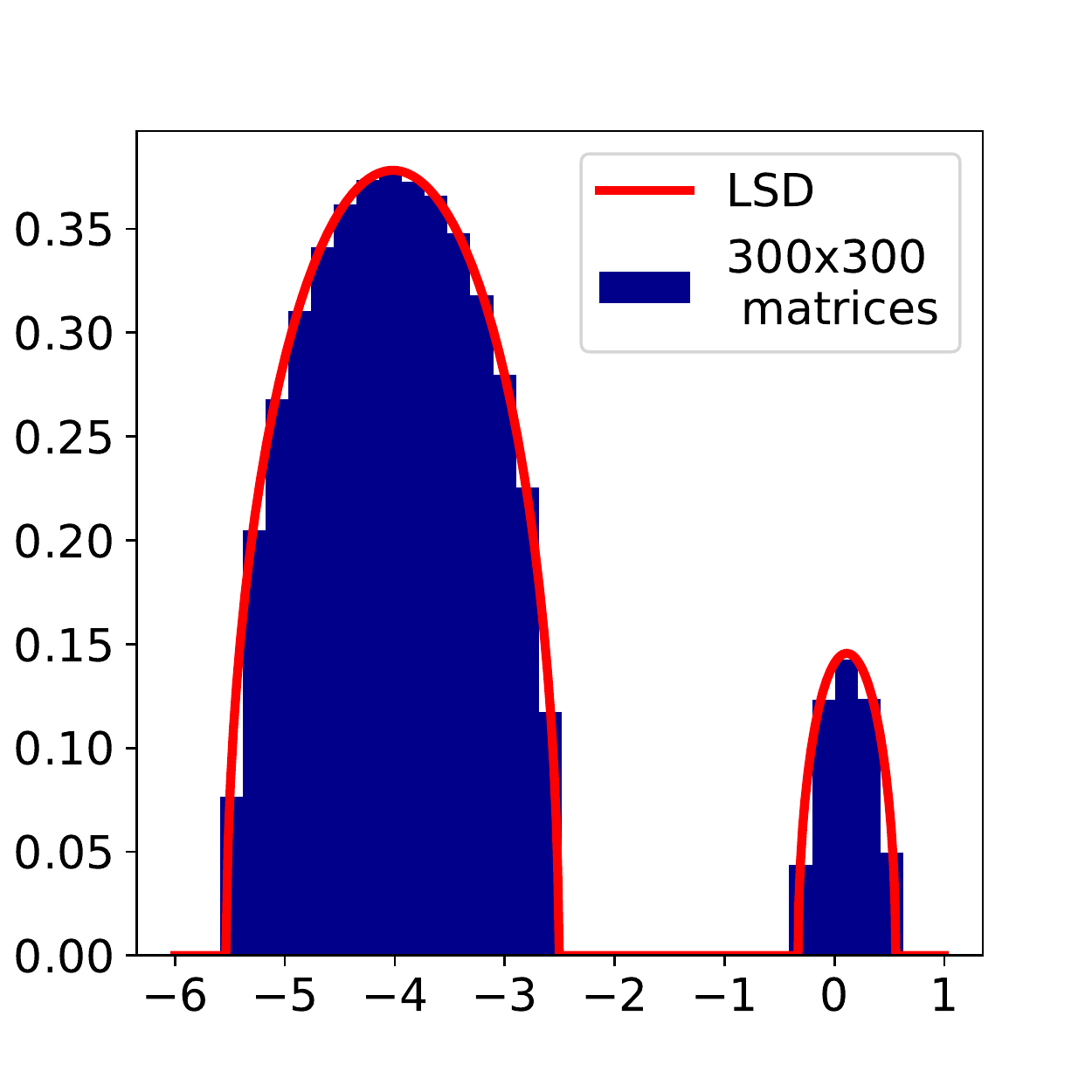}} 
\end{tabular}
\caption{Example spectra of $H'$ showing empirical spectra from 100 $300\times 300$ matrices and the corresponding LSDs computed from (\ref{eq:master_quartic}). Here $b=b_1=1$, $\kappa=0.9$, $\sigma_z$=1 and $x_1$ is varied to give the three different behaviours.}
\label{fig:spectra}
\end{figure}

\section{The asymptotic complexity}\label{sec:complexity}
In the previous section, we have found the equilibrium measure, $\mu_{eq}$, of the ensemble of random matrices \begin{align}
    H' = bM + b_1\left(\begin{array}{cc}
     M_1 & 0 \\ 0 & 0\end{array}\right) - x_1\left(\begin{array}{cc}
     I & 0 \\ 0 & 0\end{array}\right), ~~ M\sim GOE^N, ~ M_1\sim GOE^{\kappa N}.
\end{align}
The Coulomb gas approximation gives us a method of computing $\mathbb{E} |\det(H'-x)|$:
\begin{align}
    \mathbb{E} |\det(H'-x)| \approx \exp\left\{N\int \log|z - x| d\mu_{eq}(z)\right\}.\label{eq:quadrature}
\end{align}
We have access to the density of $\mu_{eq}$ pointwise (in $x$ and $x_1$) numerically, and so (\ref{eq:quadrature}) is a matter of one-dimensional quadrature. Recalling (\ref{eq:ecn}), we then have \begin{align}
   \mathbb{E} C_N \approx  K_N'\iint_B dxdx_1 ~ \exp\left\{-(N-2)\left(  \frac{1}{2s^2}x^2 + \frac{1}{2s_1^2} (x_1)^2 - \int\log|z - x| d\mu_{eq}(z) \right)\right\}\equiv K_N'\iint_B dxdx_1 ~ e^{-(N-2)\Phi(x, x_1)}
\end{align}

 where \begin{align}
     K_N' = K_N \sqrtsign{\frac{N-2}{2\pi s_1^2}} \sqrtsign{\frac{N-2}{2\pi s^2}}.
 \end{align}
Due to Lemma \ref{lemma:kn}, the constant term has asymptotic form \begin{align}
     \frac{1}{N}\log K_N' \sim \frac{1}{2}\log{2} + \frac{1}{2}\log{\pi} - \frac{\kappa}{2}\log\left(p + \sigma_z^22^{p+q}(p+q)\right) - \frac{\kappa'}{2} \log\left(\sigma_z^2(p+q) 2^{p+q}\right) - \frac{\kappa}{2}\log\kappa - \frac{\kappa'}{2}\log\kappa' \equiv K
 \end{align}
We then define the desired $\Theta(u_D, u_G)$ as \begin{align}
    \lim \frac{1}{N} \log\mathbb{E}C_N = \Theta(u_D, u_G)
\end{align}
and we have \begin{align}
    \Theta(u_D, u_G) = K - \min_B \Phi.
\end{align}
Using these numerical methods, we obtain the plot of $\Phi$ in $B$ and a plot of $\Theta$ for some example $p,q,\sigma_z, \kappa$ values, shown in Figures \ref{fig:phi_indom_latest}, \ref{fig:theta_latest}. Numerically obtaining the maximum of $\Phi$ on $B$ is not as onerous as it may appear, since $-\Phi$ grows quadratically in $|x|, |x_1|$ at moderate distances from the origin.

\begin{figure}[h]
    \centering
 \includegraphics[width=0.33\textwidth]{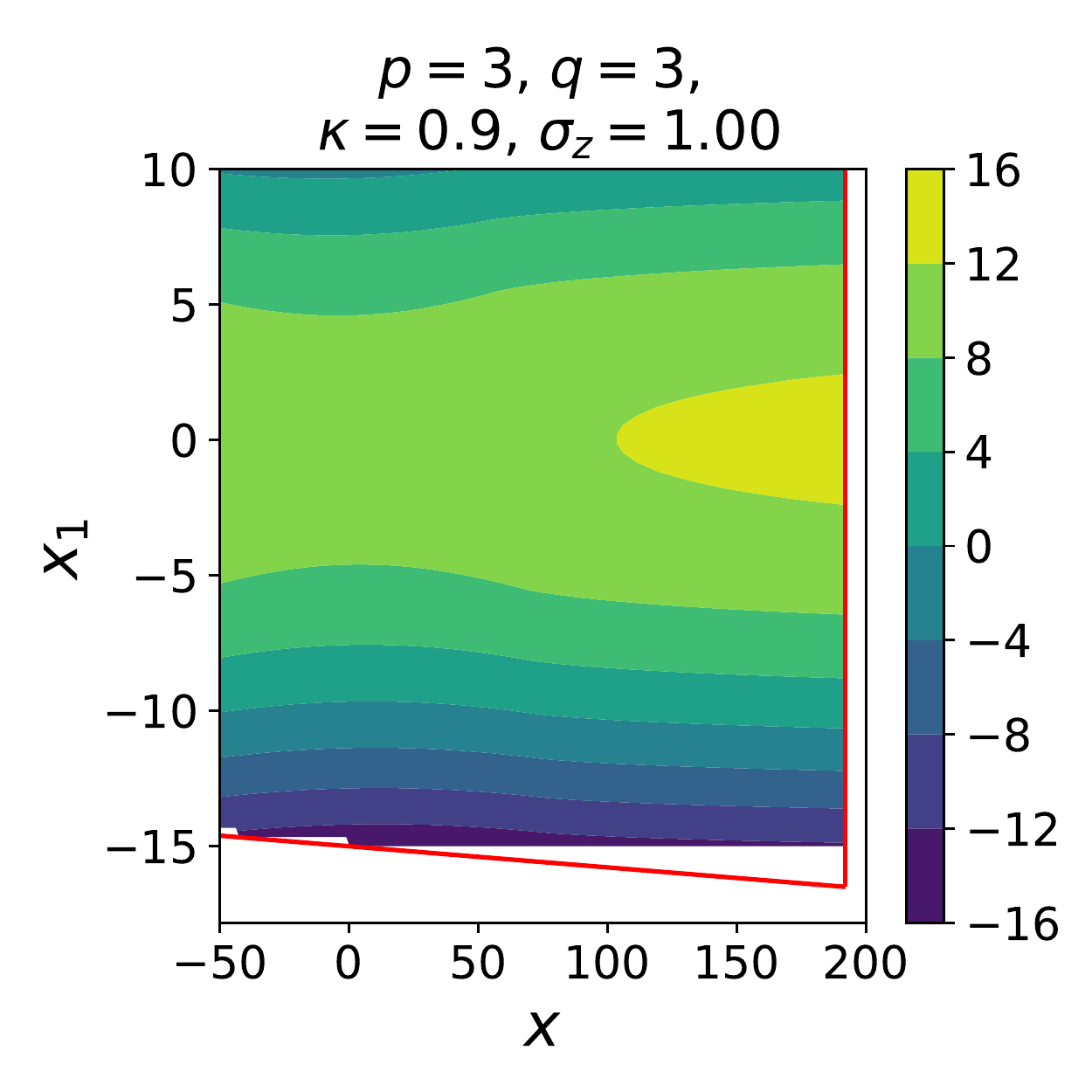}
    \caption{$\Phi$ for  $p=q=3, \sigma_z=1, \kappa=0.9$. Red lines show the boundary of the integration region $B$.}
    \label{fig:phi_indom_latest}
\end{figure}

\begin{figure}[h]
    \centering
    \includegraphics[width=\textwidth]{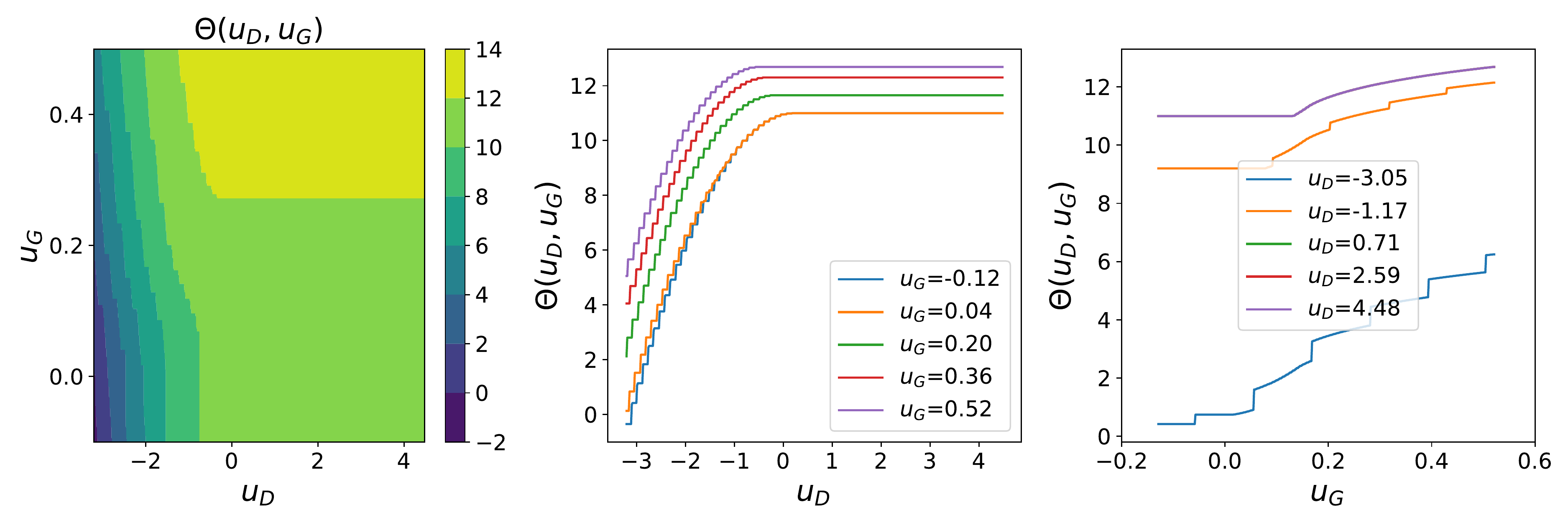}
    \caption{$\Theta$ and its cross-sections, fixing separately $u_D$ and $u_G$. Here $p=q=3, \sigma_z=1, \kappa=0.9$.}
    \label{fig:theta_latest}
\end{figure}

We numerically verify the legitimacy of this Coulomb point approximation with Monte Carlo integration \begin{align}
    \mathbb{E}|\det(H'-x)| \approx \frac{1}{n}\sum_{i=1}^n \prod_{j=1}^N |\lambda_j^{(i)} - x|,\label{eq:mc_det}
\end{align}
where $\lambda^{(i)}_j$ is the $j$-th eigenvalues of the $i$-th i.i.d. sample from the distribution of $H'$. The results, comparing $N^{-1}\log\mathbb{E}|\det(H'-x)|$ at $N=50$ for a variety of $x,x_1$ are show in Figure \ref{fig:mc_coulomb_verify}. Note the strong agreement even at such modest $N$, however to rigorously substantiate the Coulomb gas approximation in (\ref{eq:quadrature}), we must prove a concentration result.

\begin{figure}[h]
    \centering
    \includegraphics[width=\textwidth]{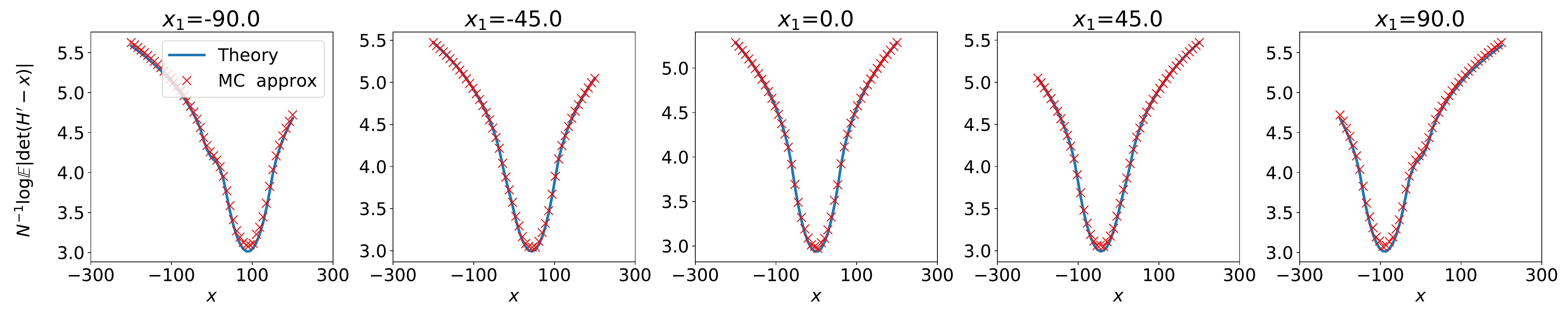}    
    \caption{Comparison of (\ref{eq:quadrature}) and (\ref{eq:mc_det}), verifying the Coulomb gas approximation numerically. Here $p=q=3, \sigma_z=1, \kappa=0.9$. Sampled matrices for MC approximation are dimension $N=50$, and $n=50$ MC samples have been used.}
    \label{fig:mc_coulomb_verify}
\end{figure}

\begin{lemma}\label{lemma:coulomb_det}
Let $(H_N)_{N=1}^{\infty}$ be a sequence of random matrices, where for each $N$ \begin{align}
    H_N \overset{d}{=}  bM + b_1\left(\begin{array}{cc}
     M_1 & 0 \\ 0 & 0\end{array}\right) - x_1\left(\begin{array}{cc}
     I & 0 \\ 0 & 0\end{array}\right)
\end{align}
and $M\sim GOE^N$, $M_1\sim GOE^{\kappa N}$. Let $\mu_N$ be the empirical spectral measure of $H_N$ and say $\mu_N\rightarrow \mu_{eq}$ weakly almost surely. Then for any $(x, x_1)\in\mathbb{R}^2$\begin{align}
\mathbb{E} |\det(H_N-xI)| = \exp\left\{N (1 + o(1))\int \log|z - x| d\mu_{eq}(z)\right\}
\end{align}
 as $N\rightarrow\infty$.
\end{lemma}

\begin{proof}
We begin by establishing an upper bound. For any $\alpha,\beta>0$ \begin{align}
    \expect |\det (H_N-xI)| &= \expect\left[\exp\left\{N\int \log|z-x| d\mu_N(z)
    \right\}\right] \notag\\
    & \leq \underbrace{\left(\expect\left[\exp\left\{2N\int \max\left(-\alpha, \min\left(\log|x-z|, \beta\right)\right) d\mu_N(z)\right\}\right]\right)^{1/2}}_{A_N}\notag\\ & ~~~~~~~ \underbrace{\left(\expect\left[\exp\left\{2N\int \log|x-z|\indic\{|x-z|\geq e^{\beta}\}d\mu_N(z)\right\}\right]\right)^{1/2}}_{B_N}.
\end{align}
Considering $B_N$, we have \begin{align}
    \log|x-z| \indic\{|x-z| \geq e^{\beta}\} \leq |x-z|^{1/2}\indic\{|x-z| \geq e^{\beta}\} \leq e^{-\beta/2}|x-z|
\end{align}
and so \begin{align}
  \expect\left[\exp\left\{2N\int \log|x-z|\indic\{|x-z|\geq e^{\beta}\}\right\}\right] &\leq  \expect \left[\exp\left\{2N e^{-\beta/2} \frac{\Tr |H_N - xI|}{N}\right\}\right] \notag\\
   &= \expect \left[\exp\left\{2e^{-\beta/2} \Tr |H_N - xI|\right\}\right].
\end{align}
% The entries of $H_N$ are Gaussian distributed with variance of order $N^{-1}$ and mean at most $\mathcal{O}(|x_1|)$. Therefore \begin{align}
%     \mathbb{P}\left(N^{-1} |H_N - xI|_{ij} > 1 + \max\{|x|, |x_1|\} \right) \lesssim e^{-\Omega(N^2)}
% \end{align}
% where the right hand side is independent of $|x|, |x_1|$. It follows that
% \begin{align}
%  \mathbb{P}\left( N^{-1}\Tr|H_N - xI| > 1 + \max\{|x|, |x_1|\} \right) &= 1 - \mathbb{P}\left( N^{-1}\Tr|H_N - xI| \leq 1 + \max\{|x|, |x_1|\} \right) \notag\\
% \end{align}
The entries of $H_N$ are Gaussians with variance $\frac{1}{N}b^2, \frac{1}{2N}b^2, \frac{1}{N}(b^2 + b_1^2)$ or $\frac{1}{2N}(b^2+b_1^2)$ and all the diagonal and upper diagonal entries are independent. All of these variances are $\mathcal{O}(N^{-1})$, so
\begin{align}
    |H_N - x|_{ij} \leq |x| + |x_1| + \mathcal{O}(N^{-1/2})|X_{ij}|
\end{align}
where the $X_{ij}$ are i.i.d. standard Gaussians for $i\leq j$. It follows that \begin{align}
       \expect \left[\exp\left\{2e^{-\frac{\beta}{2}} \Tr |H_N- xI|\right\}\right]   &\leq e^{2e^{-\frac{\beta}{2}}N(|x| + |x_1|)}\expect_{X\sim\mathcal{N}(0,1)} e^{2e^{-\frac{\beta}{2}}\mathcal{O}(N^{1/2})|X|}.
\end{align}
Elementary calculations give \begin{align}
    \expect_{X\sim \mathcal{N}(0,1)} e^{c |X|} \leq \frac{1}{2}\left(e^{-c^2} + e^{c^2}\right) \leq e^{c^2}
\end{align}
and so \begin{align}
      \expect \left[\exp\left\{2e^{-\frac{\beta}{2}} \Tr |H_N- xI|\right\}\right]   &\leq e^{2e^{-\frac{\beta}{2}}N(|x| + |x_1|)} e^{4e^{-\beta} \mathcal{O}(N)} = \exp\left\{2N\left(e^{-\frac{\beta}{2}}(|x| + |x_1|) + e^{-\beta}\mathcal{O}(1)\right)\right\}
\end{align}
thus when we take $\beta\rightarrow \infty$, we have $B_N \leq e^{o(N)}$.

% \begin{align}
%       \mathbb{P}\left( N^{-1}\Tr|H_N - xI| > 1 + \max\{|x|, |x_1|\} \right) \leq   \mathbb{P}\left(N^{-1} |H_N - xI|_{11} > 1 + \max\{|x|, |x_1|\} \right)\lesssim e^{-\Omega(N^2)}
% \end{align} and so \begin{align}
%   \expect \left[\exp\left\{2e^{-\beta/2} \Tr |H_N- xI|\right\}\right] &\leq  \exp\left\{2Ne^{-\beta/2} (1 + \max\{|x|, |x_1|\} )\right\} \left(1 - e^{-\Omega(N^2)}\right) + ...
% \end{align}
% which becomes unity in the limit $\beta\rightarrow\infty$, showing that $B_N \leq 1$.
Considering $A_N$, it is sufficient now to show \begin{align}
    \expect \left[\exp\left\{ 2N \int f(z) d\mu_N(z)\right\}\right] = \exp\left\{2N\left(\int f(z) d\mu_{eq}(z) + o(1)\right)\right\}
\end{align}
where $f(z) = 2\max\left(\min(\log|x-z|, \beta), -\alpha\right)$, a continuous and bounded function. For any $\epsilon>0$, we have \begin{align}
    \expect \left[\exp\left\{2N\int f(z) d\mu_N(z)\right\}\right] \leq \exp\left\{2N\left(\int f(z) d\mu_{eq}(z) + \epsilon\right)\right\} + e^{2N||f||_{\infty}}\mathbb{P}\left(\int f(z) d\mu_N(z) \geq \int f(z)d\mu_{eq}(z) + \epsilon\right).
\end{align}
The entries of $H_N$ are Gaussian with $\mathcal{O}(N^{-1})$ variance and so obey a log-Sobolev inequality as required by Theorem 1.5 from \cite{guionnet2000concentration}. The constant, $c$, in the inequality is independent of $N, x, x_1$, so we need not compute it exactly. The theorem from \cite{guionnet2000concentration} then gives
\begin{align}
    \mathbb{P}\left(\int f(z) d\mu_N(z) \geq \int f(z)d\mu_{eq}(z) + \epsilon\right) \leq \exp\left\{-\frac{N^2}{8c}\epsilon^2\right\}.
\end{align}
We have shown \begin{align}
  \expect|\det(H_N - xI)| \leq  A_NB_N \leq \exp\left\{N(1 + o(1))\left(\int f(z)d\mu_{eq}(z)\right)\right\} \leq  \exp\left\{N(1 + o(1))\left(\int \log|x-z|d\mu_{eq}(z)\right)\right\}.\label{eq:lemma_upper_bnd}
\end{align}

We now need to establish a complimentary lower bound to complete the proof. By Jensen's inequality \begin{align}
    \expect |\det(H_N-x)| &\geq \exp\left(N\mathbb{E}\left[\int \log|z-x| d\mu_N(z)\right]\right) \notag\\
    & \geq \exp\left(N\mathbb{E}\left[\int\max\left(-\alpha,  \log|z-x|\right) d\mu_N(z)\right]\right)  \exp\left(N\expect\left[\int \log|z-x| \indic\{|z-x| \leq e^{-\alpha}\}d\mu_N(z)\right]\right)\notag\\
    &\geq \exp\left(N\mathbb{E}\left[\int\min\left(\beta, \max\left(-\alpha,  \log|z-x|\right)\right) d\mu_N(z)\right]\right)  \exp\left(N\expect\left[\int \log|z-x| \indic\{|z-x| \leq e^{-\alpha}\}d\mu_N(z)\right]\right)\label{eq:ldp_lower_bound_2_terms}
\end{align}
for any $\alpha, \beta >0$. Convergence in law of $\mu_N$ to $\mu_{eq}$ and the dominated convergence theorem give \begin{align}
    \exp\left(N\mathbb{E}\left[\int\min\left(\beta, \max\left(-\alpha,  \log|z-x|\right)\right) d\mu_N(z)\right]\right) \geq \exp\left\{N \left(\int \log|x-z| d\mu_{eq}(z) + o(1)\right)\right\}
\end{align}
for large enough $\beta$, because $\mu_{eq}$ has compact support. It remains to show that the expectation inside the exponent in the second term of (\ref{eq:ldp_lower_bound_2_terms}) converges to zero uniformly in $N$ in the limit $\alpha \rightarrow \infty$. \\

By (\ref{eq:stieljes}), it is sufficient to consider $\langle G_N(z)\rangle$, which is computed via (\ref{eq:lsd_pre_saddle}). Let us define the function $\Psi$ so that \begin{align}
    \langle G_N(z) \rangle = \frac{2}{b^2} \int d\sigma d\sigma[1] (z-i\sigma_{FF}) e^{-N\Psi(\sigma, \sigma[1])}.
\end{align}
Henceforth, $\sigma_{FF}^*, \sigma_{FF}[1]^*, \sigma_{BB}^*, \sigma_{BB}[1]^*$ are the solution to the saddle point equations (\ref{eq:boson_saddle1}-\ref{eq:ferm_saddle2}) and  $\tilde{\sigma}_{FF}, \tilde{\sigma}_{FF}[1], \tilde{\sigma}_{BB}, \tilde{\sigma}_{BB}[1]$ are integration variables. Around the saddle point \begin{align}
    z - i\sigma_{FF} = z - i\sigma_{FF}^* - iN^{-\frac{1}{r}}\tilde{\sigma}_{FF}
\end{align}
for some $r\geq 2$. We use the notation $\vec{\sigma}$ for $(\sigma_{BB}, \sigma_{BB}[1], \sigma_{FF}, \sigma_{FF}[1])$ and similarly $\vec{\sigma}_{BB}, \vec{\sigma}_{FF}$. A superscript asterisk on $\Psi$ or any of its derivatives is short hand for evaluation at the saddle point. While the Hessian of $\Psi$ may not in general vanish at the saddle point, \begin{align}
    \int d\tilde{\sigma}d\tilde{\sigma}[1] \tilde{\sigma}_{FF} e^{-N \tilde{\vec{\sigma}}^T \nabla^2 \Psi^* \tilde{\vec{\sigma}}} = 0
\end{align}
and so we must go to at least the cubic term in the expansion of $\Psi$ around the saddle point, i.e. \begin{align}
      \langle G_N(z) \rangle = G(z) -  \frac{2i}{b^2 N^{5/3}}\underbrace{\int_{-\infty}^{\infty} d\tilde{\vec{\sigma}}_{BB} d\tilde{\vec{\sigma}}_{FF} \tilde{\sigma}_{FF} e^{-\frac{1}{6} \tilde{\sigma}^i \tilde{\sigma}^j \tilde{\sigma}^k \partial_{ijk}\Psi^*}}_{E(z; x_1)} + \text{ exponentially smaller terms}.
\end{align}
The bosonic (BB) and fermionic (FF) coordinates do not interact, so we can consider derivatives of $\Phi$ as block tensors. Simple differentiation gives \begin{align}
    (\nabla\Psi)_B &= \left(\begin{array}{c}
        \frac{2}{b^2}\sigma_{BB} - \kappa\left(\sigma_{BB} + \sigma_{BB}[1] + z + x_1\right)^{-1} - \kappa'\left(\sigma_{BB} + z\right)^{-1}  \\
        \frac{2}{b_1^2}\sigma_{BB}[1] - \left(\sigma_{BB} + \sigma_{BB}[1] + z + x_1\right)^{-1}\end{array}\right)\notag\\
   \implies  (\nabla^2\Psi)_B &= \left(\begin{array}{cc}
       \kappa\left(\sigma_{BB} + \sigma_{BB}[1] + z + x_1\right)^{-2} + \kappa'\left(\sigma_{BB} + z\right)^{-2}  &
       \kappa\left(\sigma_{BB} + \sigma_{BB}[1] + z + x_1\right)^{-2}\\
     \left(\sigma_{BB} + \sigma_{BB}[1] + z + x_1\right)^{-2} &
      \left(\sigma_{BB} + \sigma_{BB}[1] + z + x_1\right)^{-2}
       \end{array}\right)\\
 \implies (\nabla^3\Psi)_B^* &= \left(\left(\begin{array}{cc}
      A_B\kappa + B_B\kappa' & A_B\kappa \\
      A_B &  A_B
 \end{array}\right), A_B\left(\begin{array}{cc}
      \kappa & \kappa \\
      1 &  1
 \end{array}\right) \right),     
\end{align}
where \begin{align}
    A_B = -\frac{2}{\left(\sigma_{BB}^* + \sigma_{BB}^*[1] + z + x_1\right)^3}, ~~~ B_B= -\frac{2}{\left(\sigma_{BB}^* + z\right)^3}.
\end{align}
$(\nabla^3\Psi)_F^*$ follows similarly with \begin{align}
        A_F = -\frac{2}{\left(\sigma_{FF}^* + \sigma_{FF}^*[1] - iz - ix_1\right)^3}, ~~~ B_F= -\frac{2}{\left(\sigma_{FF}^* - iz\right)^3}.
\end{align}
By the saddle point equations (\ref{eq:boson_saddle1})-(\ref{eq:ferm_saddle2}) we have \begin{align}
    A_B &= 2(\sigma_{BB}[1]^*)^3, ~~ B_B = \frac{2}{(\kappa')^3}\left(\frac{2\kappa}{b_1^2}\sigma_{BB}[1]^* - \frac{2}{b^2} \sigma_{BB}^*\right)^3\label{eq:ab_bb_final}\\
        A_F &= 2(\sigma_{FF}[1]^*)^3, ~~ B_F = \frac{2}{(\kappa')^3}\left(\frac{2\kappa}{b_1^2}\sigma_{FF}[1]^* - \frac{2}{b^2} \sigma_{FF}^*\right)^3.\label{eq:af_bf_final}
\end{align}
Let $\xi_1= \tilde{\sigma}_{BB}, \xi_2 =\tilde{\sigma}_{BB}[1]$. Then \begin{align}
   ( \tilde{\sigma}^i \tilde{\sigma}^j \tilde{\sigma}^k \partial_{ijk}\Phi^*)_B &= \left(A_B\kappa + B_B\kappa'\right)\xi_1^3 + A_B(2\kappa +1 ) \xi_1^2\xi_2[1] + A_B(\kappa +2 ) \xi_1\xi_2^2 + A_B\xi_2^3\notag\\
   &= A_B\left[\xi_2^3 + (2\kappa +1 )\xi_2\xi_1^2 +(2+ \kappa)\xi_1\xi_2^2 + C\xi_1^3\right] + \left(B_B\kappa' + A_B\kappa - CA_B\right)\xi_1^3
\end{align}
for any $C$. Let $\xi_1 = a_1\xi_1'$ and then choose $C = a_1^{-3}$ and $a_1 = (2+\kappa)(2\kappa + 1)^{-1}$ to give \begin{align}
     ( \tilde{\sigma}^i \tilde{\sigma}^j \tilde{\sigma}^k \partial_{ijk}\Phi^*)_B &= A_B(\xi_1' + \xi_2)^3 + (B_B\kappa' + A_B\kappa - CA_B)a_1^3(\xi_1')^3 \equiv A_B\eta^3 + D_B\xi^3
\end{align}
with $\eta = \xi_1' + \xi_2$, $\xi=\xi_1'$, $D_B=B_B\kappa' + A_B\kappa - a_1^{-3}A_B$. The expressions for $ ( \tilde{\sigma}^i \tilde{\sigma}^j \tilde{\sigma}^k \partial_{ijk}\Phi^*)_F$ follow identically. We thus have \begin{align}
    E(z;x_1) \propto \left(\int_0^{\infty} d\xi ~ \xi \int_{\xi}^{\infty}d\eta ~ e^{A_F\eta^3 + D_F\xi^3}\right)\left(\int_0^{\infty} d\xi ~  \int_{\xi}^{\infty}d\eta ~ e^{A_B\eta^3 + D_B\xi^3}\right)
\end{align} 
or perhaps with the the integration ranges reversed depending on the signs of $\Re A_F, \Re A_B, \Re D_F, \Re D_B$. We have \begin{align}
    |E(z;  x_1)| & \leq  \left|\int_0^{\infty} d\xi ~ \xi \int_{\xi}^{\infty}d\eta ~ e^{A_F\eta^3 + D_F\xi^3}\right|\cdot\left|\int_0^{\infty} d\xi ~  \int_{\xi}^{\infty}d\eta ~ e^{A_B\eta^3 + D_B\xi^3}\right|\notag\\
    & \leq  \int_0^{\infty} d\xi ~ \xi \int_{\xi}^{\infty}d\eta ~| e^{A_F\eta^3 + D_F\xi^3}|\cdot\int_0^{\infty} d\xi ~  \int_{\xi}^{\infty}d\eta ~| e^{A_B\eta^3 + D_B\xi^3}|\notag\\
    & \leq  \int_0^{\infty} d\xi ~ \xi \int_{0}^{\infty}d\eta ~| e^{A_F\eta^3 + D_F\xi^3}|\cdot\int_0^{\infty} d\xi ~  \int_{0}^{\infty}d\eta ~| e^{A_B\eta^3 + D_B\xi^3}|\notag\\
    & \leq \left(|\mathfrak{M} D_F|\right)^{-2/3}\left(|\mathfrak{M} A_F|\right)^{-1/3}
    \left(|\mathfrak{M} D_B|\right)^{-1/3}
    \left(|\mathfrak{M} A_B|\right)^{-1/3}\left(\int_0^{\infty} e^{-\xi^3}d\xi\right)^3 \left(\int_0^{\infty}~ \xi e^{-\xi^3}d\xi\right)\label{eq:esd_error_bound}
\end{align}
where we have defined \begin{align}
    \mathfrak{M} y = \begin{cases} \Re y ~~ &\text{if } \Re y \neq 0, \\
    \Im y ~~ &\text{if } \Re y = 0. \\
    \end{cases}
\end{align}
This last bound follows from a standard Cauchy rotation of integration contour if any of $D_F, A_F, D_B, A_B$ has vanishing real part.
(\ref{eq:esd_error_bound}) is valid for $D_B, A_B, D_F, A_F \neq 0$, but if $D_B=0$ and $A_B\neq 0$, then the preceding calculations are simplified and we still obtain an upper bound but proportional to $(|\mathfrak{M} A_B|)^{-1/3}$. Similarly with $A_B=0$ and $D_B\neq 0$ and similarly for $A_F, D_F$. The only remaining cases are $A_B = D_B =0$ or $A_F = D_F =0$. But recall (\ref{eq:af_bf_final}) and (\ref{eq:ferm_saddle1})-(\ref{eq:ferm_saddle2}). We immediately see that $A_F=D_F$ if and only if $\sigma_{FF}=\sigma_{FF}[1]=0$, which occurs for no finite $z, x_1$. Therefore, for \emph{fixed} $(x, x_1)\in\mathbb{R}^2$,  $\alpha > 0$ and any $z\in (x-e^{-\alpha}, x + e^{-\alpha})$ \begin{align}
    |\mathbb{E}\mu_N(z) - \mu_{eq}(z; x_1) | \lesssim N^{-5/3} C(x_1, |x| + e^{-\alpha})
\end{align}
where $C(|x_1|, |x| + e^{-\alpha})$ is positive and is decreasing in $\alpha$. Since $\mu_{eq}$ is bounded, it follows that $\mathbb{E}\mu_N$ is bounded, and therefore \begin{align}
 \mathbb{E}   \int \log|z-x| \indic\{|z-x| \leq e^{-\alpha}\} d\mu_N(z) \rightarrow 0
\end{align}
as $\alpha\rightarrow\infty$ uniformly in $N$, and so the lower bound is completed.
\end{proof}

Equipped with this result, we can now prove the legitimacy of the Coulomb gas approximation in our complexity calculation. The proof will require an elementary intermediate result which has undoubtedly appeared in various places before, but we prove it here anyway for the avoidance of doubt.

\begin{lemma}\label{lem:max_entry_bound}
Let $M_N$ be a random $N\times N$ symmetric real matrix with independent centred Gaussian upper-diagonal and diagonal entries. Suppose that the variances of the entries are bounded above by $cN^{-1}$ for some constant $c>0$. Then there exists some constant $c_e$ such that \begin{align}
    \expect ||M_N||_{\text{max}}^N \lesssim e^{c_eN}.
\end{align}
\end{lemma}
\begin{proof}
Let $\sigma_{ij}^2$ denote the variance of $M_{ij}$. Then
\begin{align}
    \mathbb{E}||M||_{max}^N &\leq \sum_{i,j} \expect|M_{i,j}|^N\notag\\
    &= \sum_{i,j} \expect |\mathcal{N}(0, \sigma_{ij}^2)|^N\notag\\
    &= \sum_{i,j} \sigma_{ij}^N \mathbb{E}|\mathcal{N}(0,1)|^N\notag\\
    &\leq N^2c^{N/2} N^{-N/2} \expect |\mathcal{N}(0,1)|^N.
\end{align}
Simple integration with a change of variables gives \begin{align}
    \expect |\mathcal{N}(0,1)|^N &= 2^{\frac{N+1}{2}}\Gamma\left(\frac{N+1}{2}\right)
\end{align}
and then, for large enough $N$, Stirling's formula gives \begin{align}
      \expect |\mathcal{N}(0,1)|^N &\sim 2^{\frac{N+1}{2}} \sqrtsign{\pi(N+1)} \left(\frac{N+1}{2e}\right)^{\frac{N-1}{2}}\notag\\
     & \sim 2\sqrtsign{\pi} e^{-\frac{N-1}{2}} N^{N/2} \left(\frac{N+1}{N}\right)^{N/2}\notag\\
      &\sim 2\sqrtsign{\pi e} N^{N/2}.
\end{align}
So finally \begin{align}
     \mathbb{E}||M||_{max}^N &\lesssim N^2c^{N/2} = e^{\frac{1}{2}N\log{c}+ 2\log{N}} \leq e^{\left(\frac{1}{2}\log{c} + 2\right)N},
\end{align}
so defining $c_e = \frac{1}{2}\log{2} + 2$ gives the result.
\end{proof}

\begin{theorem}
For any $x_1\in\mathbb{R}$, let $H_N$ be a random $N\times N$ matrix distributed as in the statement of Lemma \ref{lemma:coulomb_det}. Then as $N\rightarrow \infty$
 \begin{align}
&\iint_B dxdx_1 ~ \exp\left\{-N\left(\frac{1}{2s^2}x^2 + \frac{1}{2s_1^2} (x_1)^2\right)\right\}\mathbb{E}|\det(H_N(x_1) - x)|\notag\\
= 
&\iint_B dxdx_1 ~ \exp\left\{-N\left(\frac{1}{2s^2}x^2 + \frac{1}{2s_1^2} (x_1)^2 - \int\log|z - x| d\mu_{eq}(z) + o(1) \right)\right\} +o(1).
 \end{align}
\end{theorem}
\begin{proof}
Let $R > 0$ be some constant, independent of $N$. Introduce the notation $B_{\leq R} = B\cap \{\vec{z}\in\mathbb{R}^2 \mid |z|\leq R\}$, and then \begin{align}
    &\Bigg|\iint_B dxdx_1 ~ \exp\left\{-N\left(\frac{1}{2s^2}x^2 + \frac{1}{2s_1^2} (x_1)^2\right)\right\}\mathbb{E}|\det(H_N(x_1) - x)|\notag\\
    &- \iint_{B_{\leq R}} dxdx_1 ~ \exp\left\{-N\left(\frac{1}{2s^2}x^2 + \frac{1}{2s_1^2} (x_1)^2\right)\right\}\mathbb{E}|\det(H_N(x_1) - x)|\Bigg|\notag\\
    \leq & \iint_{||\vec{x}||\geq R} dxdx_1 ~ \exp\left\{-N\left(\frac{1}{2s^2}x^2 + \frac{1}{2s_1^2} (x_1)^2\right)\right\}\mathbb{E}|\det(H_N(x_1) - x)|.\label{eq:conc_integral_bound1}
\end{align}

We have the upper bound (\ref{eq:lemma_upper_bnd}) of Lemma \ref{lemma:coulomb_det} but this cannot be directly applied to (\ref{eq:conc_integral_bound1}) since the bound relies on uniformity in $x, x_1$ which can only be established for bounded $x, x_1$. We use a much cruder bound instead. First, let \begin{align}
J_N = H_N + x_1 \left(\begin{array}{cc}
     I & 0 \\
     0 & 0
\end{array}\right)
\end{align} and then \begin{align}
    |\det\left(H_N - xI\right)| \leq  ||J_N||_{\text{max}}^N \max\{|x|, |x_1|\}^N =   ||J_N||_{\text{max}}^N \exp\left(N\max\{\log|x|, \log|x_1|\}\right).
\end{align}
$J_N$ has centred Gaussian entries with variance $\mathcal{O}(N^{-1})$, so Lemma \ref{lem:max_entry_bound} applies, and we find 
\begin{align}
    \expect  |\det\left(H_N - xI\right)| \lesssim  \exp\left(N\max\{\log|x|, \log|x_1|\}\right) e^{c_e N}
\end{align}
for some constant $c_e>0$ which is independent of $x, x_1$ and $N$, but we need not compute it.

Now we have \begin{align}
       &\Bigg|\iint_B dxdx_1 ~ \exp\left\{-N\left(\frac{1}{2s^2}x^2 + \frac{1}{2s_1^2} (x_1)^2\right)\right\}\mathbb{E}|\det(H_N(x_1) - x)|\notag\\
    &- \iint_{B_{\leq R}} dxdx_1 ~ \exp\left\{-N\left(\frac{1}{2s^2}x^2 + \frac{1}{2s_1^2} (x_1)^2\right)\right\}\mathbb{E}|\det(H_N(x_1) - x)|\Bigg|\notag\\
    \lesssim & \iint_{||\vec{x}||\geq R} dxdx_1 ~ \exp\left\{-N\left(\frac{1}{2s^2}x^2 + \frac{1}{2s_1^2} (x_1)^2 - \max\{\log|x|, \log|x_1|\} - c_e\right)\right\}.
\end{align}
But, since $\mu_{eq}$ is bounded and has compact support, we can choose $R$ large enough (independent of $N$) so that \begin{align}
  \frac{1}{2s^2}x^2 + \frac{1}{2s_1^2} (x_1)^2 - \max\{\log|x|, \log|x_1|\} - c_e> L > 0\label{eq:theorem_gaussian_bound}
\end{align}
for all $(x, x_1)$ with $\sqrtsign{x^2 + x_1^2} > R$ and for some fixed $L$ independent of $N$.  Whence  \begin{align}
       &\Bigg|\iint_B dxdx_1 ~ \exp\left\{-N\left(\frac{1}{2s^2}x^2 + \frac{1}{2s_1^2} (x_1)^2\right)\right\}\mathbb{E}|\det(H_N(x_1) - x)|\notag\\
    &- \iint_{B_{\leq R}} dxdx_1 ~ \exp\left\{-N\left(\frac{1}{2s^2}x^2 + \frac{1}{2s_1^2} (x_1)^2\right)\right\}\mathbb{E}|\det(H_N(x_1) - x)|\Bigg|\notag\\
    \lesssim & N^{-1}e^{-NL} \rightarrow 0\label{eq:thm_final1}
\end{align}
as $N\rightarrow\infty$. Finally, for $x, x_1$ in $B_{\leq R}$, the result of the Lemma \ref{lemma:coulomb_det} holds uniformly in $x, x_1$, so \begin{align}
    &\iint_{B_{\leq R}} dxdx_1 ~ \exp\left\{-N\left(\frac{1}{2s^2}x^2 + \frac{1}{2s_1^2} (x_1)^2\right)\right\}\mathbb{E}|\det(H_N(x_1) - x)|\notag\\ =  &\iint_{B_{\leq R}} dxdx_1 ~ \exp\left\{-N\left(\frac{1}{2s^2}x^2 + \frac{1}{2s_1^2} (x_1)^2 - \int\log|z-x| d\mu_{eq}(z; x_1)+ o(1)\right)\right\}.\label{eq:thm_final2}
\end{align}
The result follows from (\ref{eq:thm_final1}), (\ref{eq:thm_final2}) and the triangle inequality.
\end{proof}

\subsection{Asymptotic complexity with prescribed Hessian index}
Recall the complexity defined in (\ref{eq:C_Nkk_def}): \begin{align}
    C_{N, k_D, k_G} =\Bigg|\Bigg\{ \vwD\in S^{N_D  }, \vwG\in S^{N_G } ~:~ &\nabla_D \LD = 0, \nabla_G \LG = 0, \LD\in B_D, \LG\in B_G\notag\\
    &i(\nabla_D^2 L^{(D)}) = k_D, ~i(\nabla_G^2 L^{(G)}) = k_G \Bigg\}\Bigg|.\tag{\ref{eq:C_Nkk_def}}
\end{align}

The extra Hessian signature conditions in (\ref{eq:C_Nkk_def}) enforce that both generator and discriminator are at low-index saddle points. Our method for computing the complexity $C_N$ in the previous subsection relies on the Coulomb gas approximation applied to the spectrum of $H'$. However, the Hessian index constraints are formulated in the natural Hessian matrix (\ref{eq:H_first}), but our spectral calculations proceed from the rewritten form (\ref{eq:H_rewritten}). We find however that we can indeed proceed much as in \cite{baskerville2020loss}. Recall the key Hessian matrix $\tilde{H}$ given in (\ref{eq:H_first}) by \begin{align}
   \tilde{H}=  \left(\begin{array}{cc}
      \sqrtsign{2(N_D - 1)}\sqrtsign{b^2 + b_1^2}M^{(D)} & -bG \\
       b G^T & \sqrtsign{2(N_G-1)}bM^{(G)}
    \end{array}\right)
   -\sqrtsign{N-2}x \left(\begin{array}{cc}
        -I_{N_D} & 0  \\
         0 & I_{N_G} 
    \end{array}\right)  + \sqrtsign{N-2}x_1\left(\begin{array}{cc}
        I_{N_D} & 0  \\
         0 & 0
    \end{array}\right)
\end{align}
where $M^{(D)}\sim GOE^{N_D -1}$, $M^{(G)}\sim GOE^{N_G-1}$, $G$ is $N_D - 1 \times N_G - 1$ Ginibre, and all are independent. Note that we have used (\ref{eq:x_x1_def}) to slightly rewrite (\ref{eq:H_first}).
We must address the problem of computing \begin{align}
    \mathbb{E}|\det \tilde{H}|\mathbbm{1}\left\{ i\left(\sqrtsign{\kappa}(1 + \mathcal{O}(N^{-1}))\sqrtsign{b^2 + b_1^2}M_D + \frac{x+x_1}{\sqrtsign{2}}\right) = k_D, ~ i\left(\sqrtsign{\kappa'}(1 + \mathcal{O}(N^{-1}))bM_G - \frac{x}{\sqrtsign{2}}\right) = k_G\right\}.
\end{align}

Indeed, we introduce integration variables $\vec{y}_1, \vec{y}_2, \zeta_1, \zeta_1^*, \zeta_2,\zeta_2^*$, being $(N-2)$-vectors of commuting and anti-commuting variables respectively. Use $[t]$ notation to split all vectors into the first $\kappa N -1$ and last $\kappa'N-1$ components. Let \begin{align}
    A[t] = \vec{y}_1\vec{y}_1^T + \vec{y}_2\vec{y}_2^T + \zeta_1\zeta_1^{\dagger} + \zeta_2\zeta_2^{\dagger}.
\end{align}

With these definitions, we have \cite{baskerville2020loss} \begin{align}
    |\det \tilde{H}| = (2(N-2))^{\frac{N-2}{2}} \lim_{\epsilon\searrow 0} \int d\Xi &\exp\left\{-i\sqrtsign{\kappa}(1 + \mathcal{O}(N^{-1})) \sqrtsign{b^2 + b_1^2}\Tr M^{(D)} A[1]-i\sqrtsign{\kappa'}(1 + \mathcal{O}(N^{-1})) b \Tr M^{(G)} A[2] \right\}\notag\\
    &\exp\{ \mathcal{O}(\epsilon)\}\exp\{\ldots\}
\end{align}
where $d\Xi$ is the normalised measure of the $\vec{y}_1, \vec{y}_2, \zeta_1, \zeta_1^*, \zeta_2,\zeta_2^*$ and the ellipsis represents terms with no dependence on $M^{(D)}$ or $M^{(G)}$, which we need not write down.
The crux of the matter is that we must compute \begin{align}
   & \mathbb{E}_{M^{(D)}}e^{-i\sqrtsign{\kappa} \sqrtsign{b^2 + b_1^2}\Tr M^{(D)} A[1]}\mathbbm{1}\left\{i\left(M_D + \frac{x + x_1}{\sqrtsign{\kappa}\sqrtsign{b^2 + b_1^2}}(1 + \mathcal{O}(N^{-1}))\right) = k_D\right\},\\
      & \mathbb{E}_{M^{(G)}}e^{-i\sqrtsign{\kappa'} b \Tr M^{(G)} A[2]}\mathbbm{1}\left\{i\left(M_G - \frac{x}{\sqrtsign{\kappa'}b}(1 + \mathcal{O}(N^{-1}))\right) = k_G\right\},
\end{align}
but \cite{baskerville2020loss} has performed exactly these calculations (see around (5.146) therein) and so there exist constants $K^{(D)}_U,K^{(D)}_L, K^{(G)}_U,K^{(G)}_L$ such that \begin{align}
  &K^{(D)}_L e^{-Nk_D\kappa(1 + o(1)) I_1(\hat{x}_D; \sqrtsign{2})} e^{-\frac{1}{2N}(b^2 + b_1^2)\Tr A[1]^2}  \notag\\
  \leq  &\Re\mathbb{E}_{M^{(D)}}e^{-i\sqrtsign{\kappa} \sqrtsign{b^2 + b_1^2}\Tr M^{(D)} A[1]}\mathbbm{1}\left\{i\left(M_D + \frac{x + x_1}{\sqrtsign{\kappa}\sqrtsign{b^2 + b_1^2}}(1 + \mathcal{O}(N^{-1}))\right) = k_D\right\} \notag\\
  \leq &K^{(D)}_U e^{-Nk_D\kappa(1 + o(1)) I_1(\hat{x}_D; \sqrtsign{2})} e^{-\frac{1}{2N}(b^2 + b_1^2)\Tr A[1]^2}
\end{align}
and 
\begin{align}
  &K^{(G)}_L e^{-Nk_G\kappa'(1 + o(1)) I_1(\hat{x}_G; \sqrtsign{2})} e^{-\frac{1}{2N}b^2 \Tr A[2]^2}  \notag\\
  \leq  &\Re\mathbb{E}_{M^{(G)}}e^{-i\sqrtsign{\kappa'} b\Tr M^{(G)} A[2]}\mathbbm{1}\left\{i\left(M_G - \frac{x}{\sqrtsign{\kappa'}b}(1 + \mathcal{O}(N^{-1}))\right) = k_G\right\} \notag\\
  \leq &K^{(G)}_U e^{-Nk_G\kappa'(1 + o(1)) I_1(\hat{x}_G; \sqrtsign{2})} e^{-\frac{1}{2N}b^2 \Tr A[2]^2}
\end{align}
where \begin{align}
    \hat{x}_D = -\frac{x + x_1}{\sqrtsign{\kappa}\sqrtsign{b^2 + b_1^2}}, ~~ \hat{x}_G = \frac{x}{\sqrtsign{\kappa'}b}.
\end{align}
Here $I_1$ is the rate function of the largest eigenvalue of the GOE as obtained in \cite{arous2001aging} and used in \cite{auffinger2013random, baskerville2020loss}: \begin{align}
    I_1(u; E) = \begin{cases} \frac{2}{E^2}\int_u^{-E} \sqrtsign{z^2 - E^2}dz ~ &\text{ for } u < -E,\\
        \frac{2}{E^2}\int_{E}^u \sqrtsign{z^2 - E^2}dz ~ &\text{ for } u > E,\\
    \infty &\text{ for } |u| < E.
    \end{cases}
\end{align}
Note that for $u< -E$ \begin{align}
    I_1(u; E) = -\frac{u}{E}\sqrtsign{u^2 - E^2} - \log\left(-u + \sqrtsign{u^2 - E^2}\right) + \log{E}
\end{align}
and for $u>E$ we simply have $I_1(u; E) = I_1(-u; E)$. Note also that $I_1(ru; E) = I_1(u, E/r).$

We have successfully dealt with the Hessian index indicators inside the expectation, however we need some way of returning to the form of $\tilde{H}$ in (\ref{eq:H_rewritten}) so the complexity calculations using the Coulomb gas approach can proceed as before. We can achieve this with inverse Fourier transforms:
\begin{align}
    e^{-\frac{1}{2N}(b^2 + b_1^2)\Tr A[1]^2} &= \mathbb{E}_{M_D}e^{-i\sqrtsign{\kappa}\sqrtsign{b^2 + b_1^2}\Tr M_DA[1]}\\
     e^{-\frac{1}{2N}b^2\Tr A[2]^2} &= \mathbb{E}_{M_G}e^{-i\sqrtsign{\kappa'}b\Tr M_GA[2]}   
\end{align}
from which we obtain \begin{align}
&K_Le^{-Nk_D\kappa(1 + o(1)) I_1(\hat{x}_D; \sqrtsign{2})} e^{-Nk_G\kappa'(1 + o(1)) I_1(\hat{x}_G; \sqrtsign{2})}\mathbb{E}|\det \tilde{H}|\notag\\
\leq 
  & \mathbb{E}|\det \tilde{H}|\mathbbm{1}\left\{ i\left(\sqrtsign{\kappa}(1 + \mathcal{O}(N^{-1}))\sqrtsign{b^2 + b_1^2}M_D + \frac{x+x_1}{\sqrtsign{2}}\right) = k_D, ~ i\left(\sqrtsign{\kappa'}(1 + \mathcal{O}(N^{-1}))bM_G - \frac{x}{\sqrtsign{2}}\right) = k_G\right\} \\
  \leq & K_Ue^{-Nk_D\kappa(1 + o(1)) I_1(\hat{x}_D; \sqrtsign{2})} e^{-Nk_G\kappa'(1 + o(1)) I_1(\hat{x}_G; \sqrtsign{2})} \mathbb{E}|\det \tilde{H}|.
  \end{align}
It follows that \begin{align}
&K_N'\iint_B dx dx_1 e^{-(N-2)\left[ \Phi(x, x_1) + k_G\kappa' I_1(x; \sqrtsign{2\kappa'}b) + k_D\kappa I_1\left(( - (x + x_1); \sqrtsign{2\kappa(b^2 + b_1^2)}\right)\right](1 + o(1))} \notag\\
\lesssim
    &C_{N, k_D, k_G} \notag\\
    \lesssim &K_N' \iint_B dx dx_1 e^{-(N-2)\left[ \Phi(x, x_1) + k_G\kappa' I_1(x; \sqrtsign{2\kappa'}b) + k_D\kappa I_1\left(( - (x + x_1); \sqrtsign{2\kappa(b^2 + b_1^2)}\right)\right](1 + o(1))}.
\end{align}
So we see that the relevant exponent in this case is the same as for $C_N$ but with additional GOE eigenvalue large deviation terms, giving the complexity limit \begin{align}
     \lim \frac{1}{N} \log\mathbb{E}C_{N, k_D, k_G} 
=   \Theta_{k_D, k_G}(u_D, u_G) = K - \min_B \left\{\Phi +  k_G\kappa' I_1(x; \sqrtsign{2\kappa'}b) + k_D\kappa I_1\left( - (x + x_1); \sqrtsign{2\kappa(b^2 + b_1^2)}\right)\right\}.
\end{align}
\begin{figure}
    \centering
    \includegraphics[width=\textwidth]{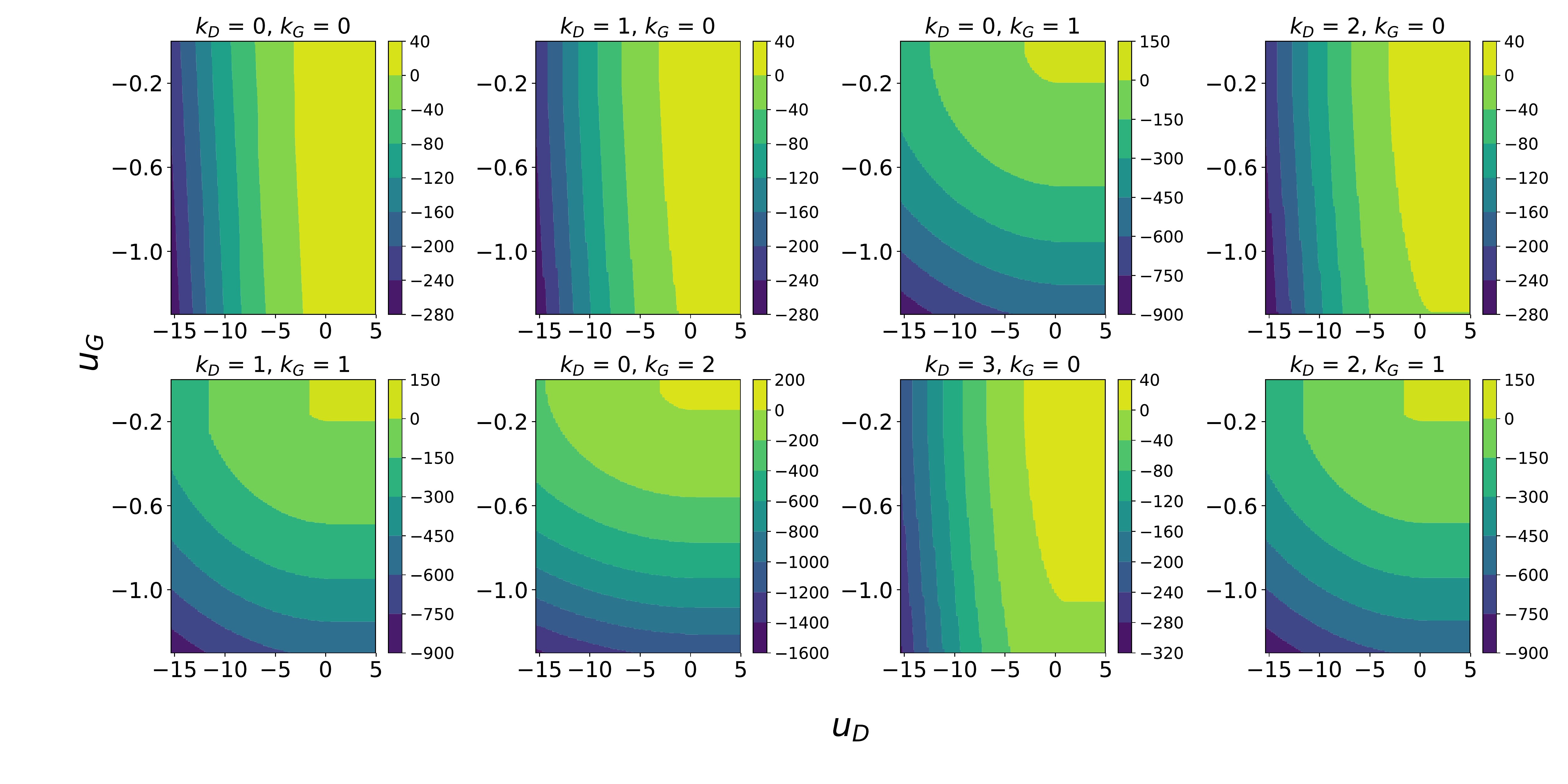}
    \caption{Contour plots of $\Theta_{k_D, k_G}$ for a few values of $k_D, k_G$. Here $p=q=3, \sigma_z=1, \kappa=0.9$.}
    \label{fig:theta_indexs}
\end{figure}
Plots of $\Theta_{k_D, k_G}$ for a few values of $k_D, k_G$ are shown in Figure \ref{fig:theta_indexs}.

\section{Implications}\label{sec:implications}
\subsection{Structure of low-index critical points}
We examine the fine structure of the low-index critical points for both spin glasses. \cite{choromanska2015loss} used the `banded structure' of low-index critical points to explain the effectiveness of gradient descent in large multi-layer perceptron neural networks. We undertake to uncover the analogous structure in our dual spin-glass model and thence offer explanations for GAN training dynamics with gradient descent. For a range of $(k_D, k_G)$ values, starting at $(0, 0)$, we compute $\Theta_{k_D, k_G}$ on an appropriate domain. In the $(u_D, u_G)$ plane, we then find the maximum $k_D$, and separately $k_G$, such that $\Theta_{k_D, k_G}(u_D, u_G)>0$. In the large $N$ limit, this procedure reveals the regions in the $(u_D, u_G)$ plane where critical points of each index of the two spin glasses are found. Figure \ref{fig:kd_kg_structure} plots these maximum $k_D, k_G$ values as contours on a shared $(u_D, u_G)$ plane. Figure \ref{fig:kd_kg_structure_filled} shows the same results, but on separate filled contour plots; the white regions in the plots clearly show the `ground state' boundary beyond which no critical points exist. We use some fixed values of the various parameters: $p=q=3, \sigma_z=1, \kappa=0.9$.

These plots reveal, unsurprisingly perhaps, that something resembling the banded structure of \cite{choromanska2015loss} is present, with the higher index critical points being limited to higher loss values for each network. The 2-dimensional analogues of the $E_{\infty}$ boundary of \cite{choromanska2015loss} are evident in the bunching of the $k_D, k_G$ contours at higher values. There is, however further structure not present in the single spin-glass multi-layer perceptron model. Consider the contour of $k_D = 0$ at the bottom of the full contour plot in Figure \ref{fig:kd_kg_structure}. Imagine traversing a path near this contour from right to left (decreasing $u_D$ values). At all points along such a path, the only critical points present are exact local minima for both networks, however the losses range over 
\begin{enumerate}[label=(\roman*)]
    \item low generator loss, high discriminator loss;
    \item some balance between generator and discriminator loss;
    \item high generator loss, low discriminator loss.
\end{enumerate}  These three states correspond qualitatively to known GAN phenomena:

\begin{enumerate}[label=(\roman*)]
    \item discriminator collapses to predicting `real' for all items;
    \item successfully trained model;
    \item generator collapses to producing garbage samples which the discriminator trivially identifies.
\end{enumerate}

\begin{figure}[h]
    \centering
    \includegraphics[width=\textwidth]{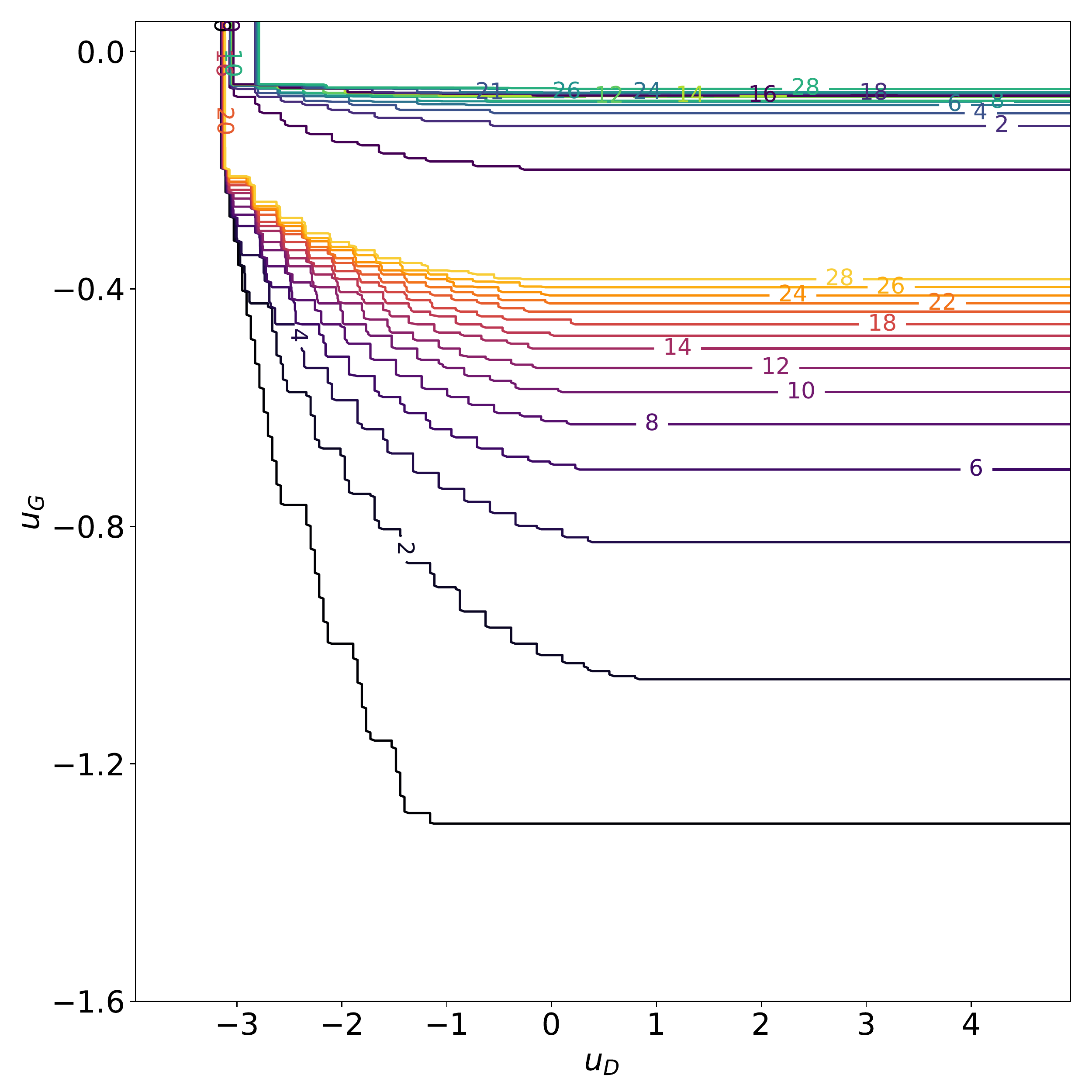}
    \caption{Contours in the $(u_D, u_G)$ plane of the maximum $k_D$ and $k_G$ such that $\Theta_{k_D, k_G}(u_D, u_G)>0$. $k_D$ results shown with a red colour red scheme, and $k_G$ with green. Only alternate contours are shown, in the interests of clarity. Here $p=q=3, \sigma_z=1, \kappa=0.9$.}
    \label{fig:kd_kg_structure}
\end{figure}
\begin{figure}[h]
    \centering
    \includegraphics[width=0.66\textwidth]{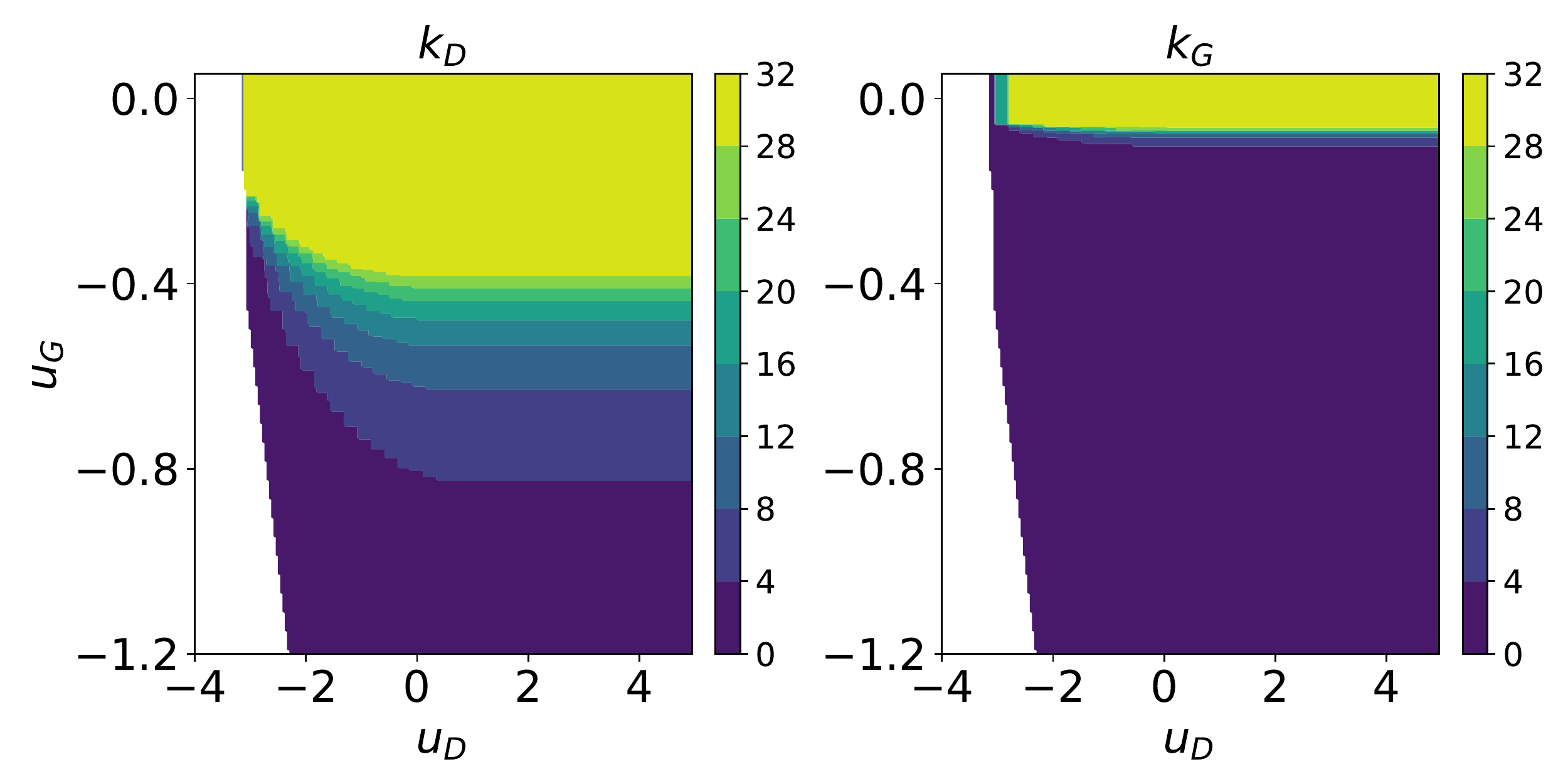}
    \caption{Contours in the $(u_D, u_G)$ plane of the maximum $k_D$ and $k_G$ such that $\Theta_{k_D, k_G}(u_D, u_G)>0$. $k_D$ results on the left, $k_G$ on the right.  Here $p=q=3, \sigma_z=1, \kappa=0.9$.}
    \label{fig:kd_kg_structure_filled}
\end{figure}

Overall, the analysis of our model reveals a loss surface that favours convergence to states of low loss for \emph{at least one of the networks}, but not necessarily both. Moreover, our plots of $\Theta$ and $\Theta_{k_D, k_G}$ in Figures \ref{fig:theta_latest}, \ref{fig:theta_indexs} demonstrate clearly the competition between the two networks, with the minimum attainable discriminator loss increasing as the generator loss decreases and vice-versa. We thus have a qualitative similarity between the minimax dynamics of real GANs and our model, including the existence of a Nash equilibrium, but also a new two-dimensional banded critical points structure. This structure offers the following explanation of large GAN training dynamics with gradient descent: \begin{enumerate}
    \item As with single feed-forward networks, the loss surface geometry encourages convergence to globally low values of at least one of the network losses.
    \item The same favourable geometry encourages convergence to successful states, where both networks achieve reasonably low loss, but also encourages convergence to failure states, where the generator's samples are too easily distinguished by the discriminator, or the discriminator has entirely failed thus providing no useful training signal to the generator.
\end{enumerate}

\subsection{Hyperparameter effects}
Our proposed model for GANs includes a few fixed hyperparameters that we expect to control features of the model, namely $\sigma_z$ and $\kappa$. Based on the results of \cite{auffinger2013random, choromanska2015loss, baskerville2020loss}, and the form of our analytical results above, we do not expect $p$ and $q$ (the number of layers in the discriminator and generator) to have any interesting effect beyond $p, q \geq 3$; this is clearly a limitation of the model. We would expect there to exist an optimal value of $\sigma_z$ that would result in minimum loss, in some sense. The effect of $\kappa$ is less clear, though we guess that, in the studied $N\rightarrow\infty$ limit, all $\kappa\in(0, 1)$ are effectively equivalent. Intuitively, $\kappa\in \{0, 1\}$ should result in the much larger network beating the smaller in the minimax game, however our results above are valid strictly for $\kappa\in (0,1)$.

\subsubsection{Effect of variance ratio}
In the definition of complexity, $u_D$ and $u_G$ are upper bounds on the loss of the discriminator and generator, respectively. We are interested in the region of the $u_D,u_G$ plane such that $\Theta(u_D, u_G)>0$, this being the region where gradient descent algorithms are expected to become trapped. We therefore investigate the minimum loss such that $\Theta > 0$, this being, for a given $\sigma_z$, the theoretical minimum loss attainable by the GAN. We consider two natural notions of loss:
\begin{enumerate}
    \item $\vartheta_D = \min\{u_D\in\mathbb{R} \mid \exists u_G\in\mathbb{R} ~:~ \Theta(u_D, u_G) > 0 \} $;
    \item $\vartheta_G = \min\{u_G\in\mathbb{R} \mid \exists u_D\in\mathbb{R} ~:~ \Theta(u_D, u_G) > 0 \} $.
    % \item $\vartheta_{D+G} = \min\{u\in\mathbb{R} \mid \exists u ~:~ u= u_D + u_G, ~ \Theta(u_D, u_G) > 0 \} $.
\end{enumerate}
We vary $\sigma_z$ over a range of values in $(10^{-5}, 10^{2})$ and compute $\vartheta_D, \vartheta_G$.

% \begin{figure}
%     \centering
%     \begin{tabular}{ccc}
%     \subfloat[$\vartheta_D$]{\includegraphics[width=0.325\linewidth]{figures/complexity/sigma_z_effect_0.pdf}} &
%         \subfloat[$\vartheta_G$]{\includegraphics[width=0.325\linewidth]{figures/complexity/sigma_z_effect_1.pdf}} &
%         \subfloat[$\vartheta_{D+G}$]{\includegraphics[width=0.325\linewidth]{figures/complexity/sigma_z_effect_2.pdf}} 
%     \end{tabular}
%     \caption{The effect of $\sigma_z$ of minimum theoretical loss. Log-scale on $x$-axis.}
%     \label{fig:varthetas}
% \end{figure}

To compare the theoretical predictions of the effect of $\sigma_z$ to real GANs, we perform a simple set of experiments. We use a DCGAN architecture \cite{radford2015unsupervised} with 5 layers in each network, using the reference PyTorch implementation from \cite{gan-code}, however we introduce the generator noise scale $\sigma_z$. For a given $\sigma_z$, we train the GANs for 10 epochs on CIFAR10 \cite{krizhevsky2009learning} and record the generator and discriminator losses. For each $\sigma_z$, we repeat the experiment 30 times and average the minimum attained generator and discriminator losses to account for random variations between runs with the same $\sigma_z$. We note that the sample variances of the loss were typically very high, despite the PyTorch random seed being fixed across all runs. We plot the sample means, smoothed with rolling averaging over a short window, in the interest of clearly visualising whatever trends are present. The results are shown in Figure \ref{fig:vary_sigma_results}.\\

There is a striking similarity between the generator plots, with a sharp decline between $\sigma_z=10^{-5}$ and around  $10^{-3}$, after which the minimum loss is approximately constant. The picture for the discriminator is less clear. Focusing on the sections $\sigma_z > 10^{-3}$, both plots show a clear minimum, at around $\sigma_z=10^{-1}$ in experiments and $\sigma_z=10^{-2}$ in theory. Note that the scales on the $y$-axes of these plots should not be considered meaningful. Though there is not precise correspondence between the discriminator curves, we claim that both theory and experiment tell the same qualitative story: increasing $\sigma_z$ to at least around $10^{-3}$ gives the lowest theoretical generator loss, and then further increasing to, tentatively, some value in $(10^{-2}, 10^{-1})$ gives the lowest possible discriminator loss at no detriment to the generator. 

% \begin{figure}
%     \centering
%     \begin{tabular}{ccc}
%     \subfloat[$\vartheta_D$]{\includegraphics[width=0.325\linewidth]{figures/realgan/ld_cifar10_gan.pdf}} &
%         \subfloat[$\vartheta_G$]{\includegraphics[width=0.325\linewidth]{figures/realgan/lg_cifar10_gan.pdf}} &
%         \subfloat[$\vartheta_{D+G}$]{\includegraphics[width=0.325\linewidth]{figures/realgan/ldlg_cifar10_gan.pdf}} 
%     \end{tabular}
%     \caption{The effect of $\sigma_z$ on the minimum losses attained when training a DCGAN on CIFAR10. Log-scale on $x$-axis.}
%     \label{fig:varthetas_real}
% \end{figure}

\begin{figure}[h]
    \centering
    \begin{tabular}{cc}
    \subfloat[Generator]{\includegraphics[width=0.49\textwidth]{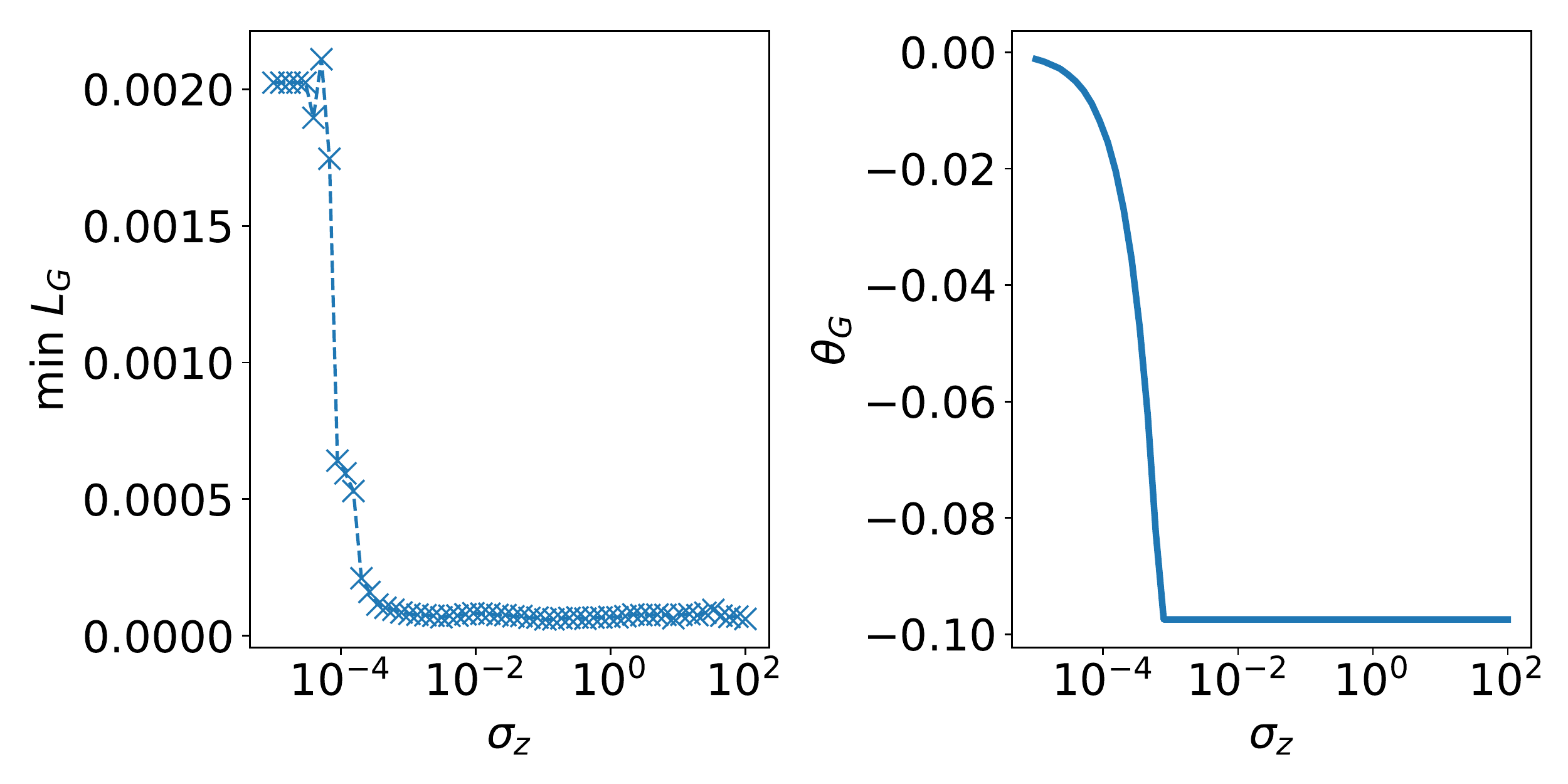}} 
    &     \subfloat[Discriminator]{\includegraphics[width=0.49\textwidth]{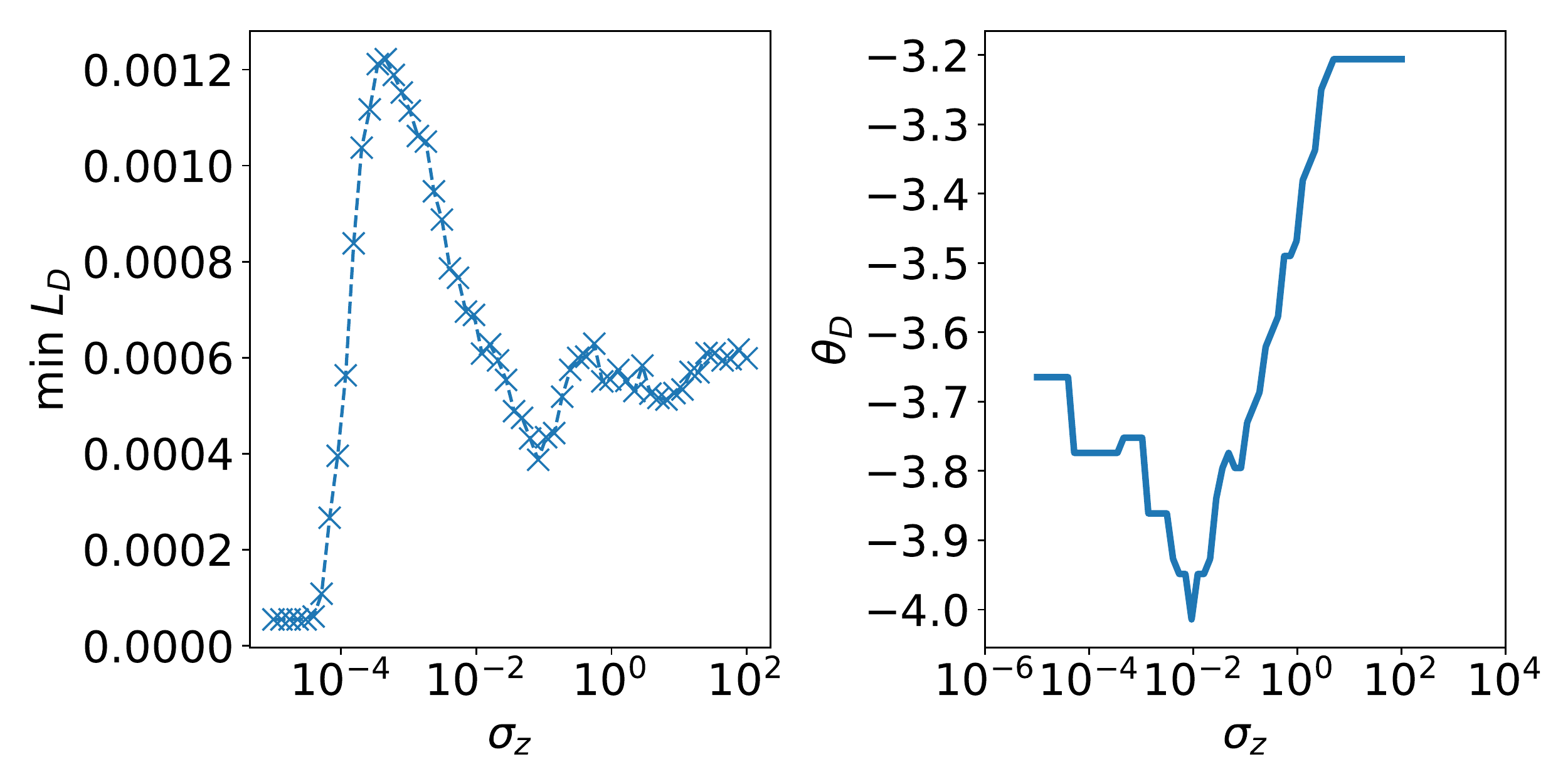}} 
    \end{tabular}
    \caption{The effect of $\sigma_z$. Comparison of theoretical predictions of minimum possible discriminator and generator losses to observed minimum losses when training DCGAN on CIFAR10. The left plots in each case (cross-dashed) show the experimental DCGAN results, and the right plots show the theoretical results $\theta_G, \theta_D$. $p=q=5$ and $\kappa=0.5$ are used in the theoretical calculations, to best match the DCGAN architecture. $\sigma_z$ is shown on a log-scale.}
    \label{fig:vary_sigma_results}
\end{figure}

\subsubsection{Effect of size ratio}
Similarly to the previous section, we can investigate the effect of $\kappa$ using $\vartheta_D, \vartheta_G$ while varying $\kappa$ over $(0,1)$. To achieve this variation in the DCGAN, we vary the number of convolutional filters in each network. The generator and discriminator are essentially mirror images of each other and the number of filters in each intermediate layer are defined as increasing functions\footnote{Number of filters in a layer is either proportional to $n_D$ or $n_D^2$ depending on the layer (and similarly with $n_G$).} of some positive integers $n_G, n_D$.  We fix $n_D + n_G=128$ and vary $n_D$ to obtain a range of $\kappa$ values, with $\kappa = \frac{n_d}{n_d + n_g}$. The results are shown in Figure \ref{fig:vary_kappa_results}.

The theoretical model predicts a a broad range of equivalently optimal $\kappa$ values centred on $\kappa=0.5$ from the perspective of the discriminator loss, and no effect of $\kappa$ on the generator loss. The experimental results similarly show a broad range of equivalently optimal $\kappa$ centred around $\kappa=0.5$, however there appear to be deficiencies in our model, particularly for higher $\kappa$ values. The results of the experiments are intuitively sensible: the generator loss deteriorates for $\kappa$ closer to 1, i.e. when the discriminator has very many more parameters than the generator, and vice-versa for small $\kappa$.

\begin{figure}[h]
    \centering
    \begin{tabular}{cc}
    \subfloat[Generator]{\includegraphics[width=0.49\textwidth]{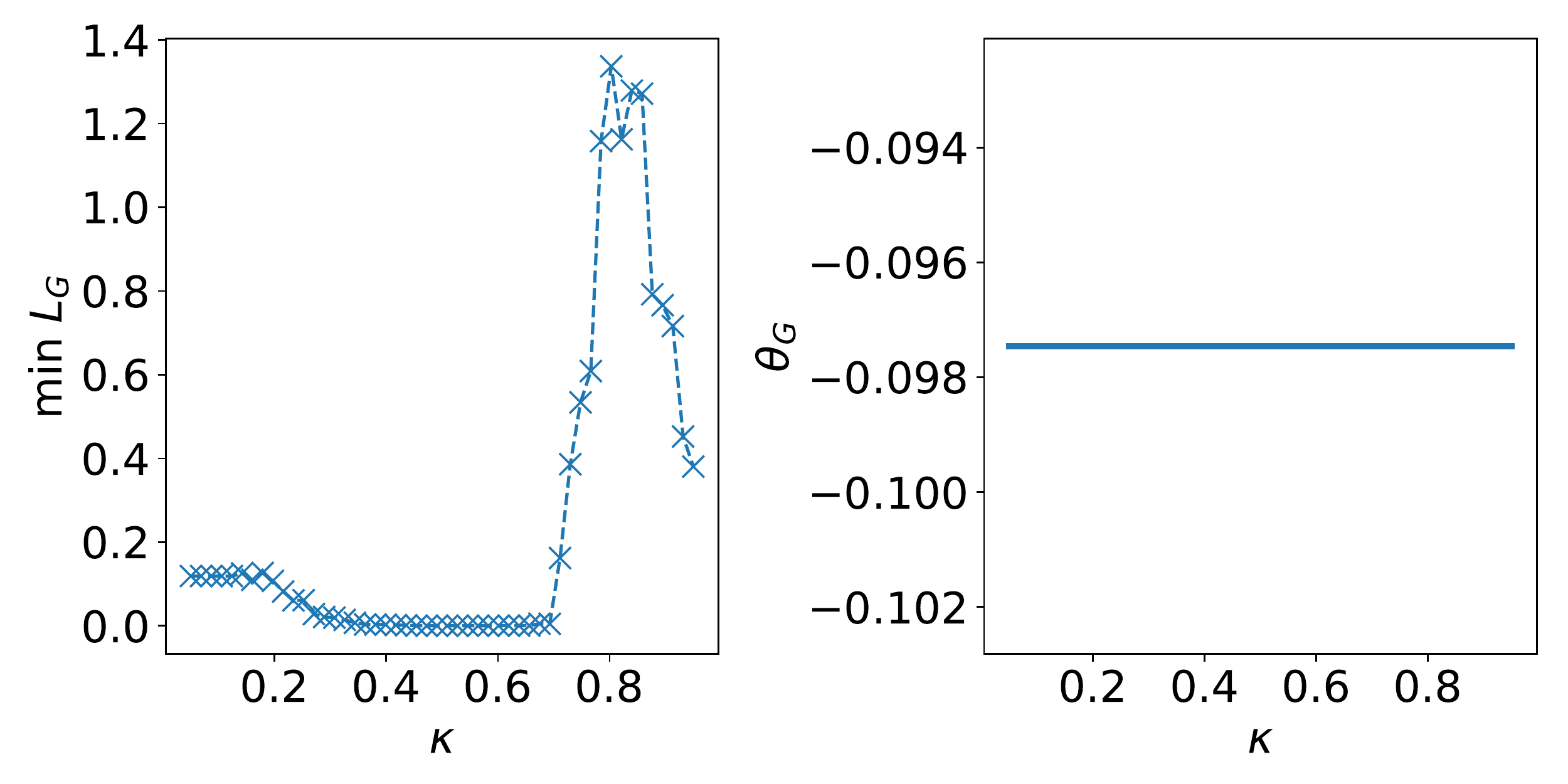}} 
    &     \subfloat[Discriminator]{\includegraphics[width=0.49\textwidth]{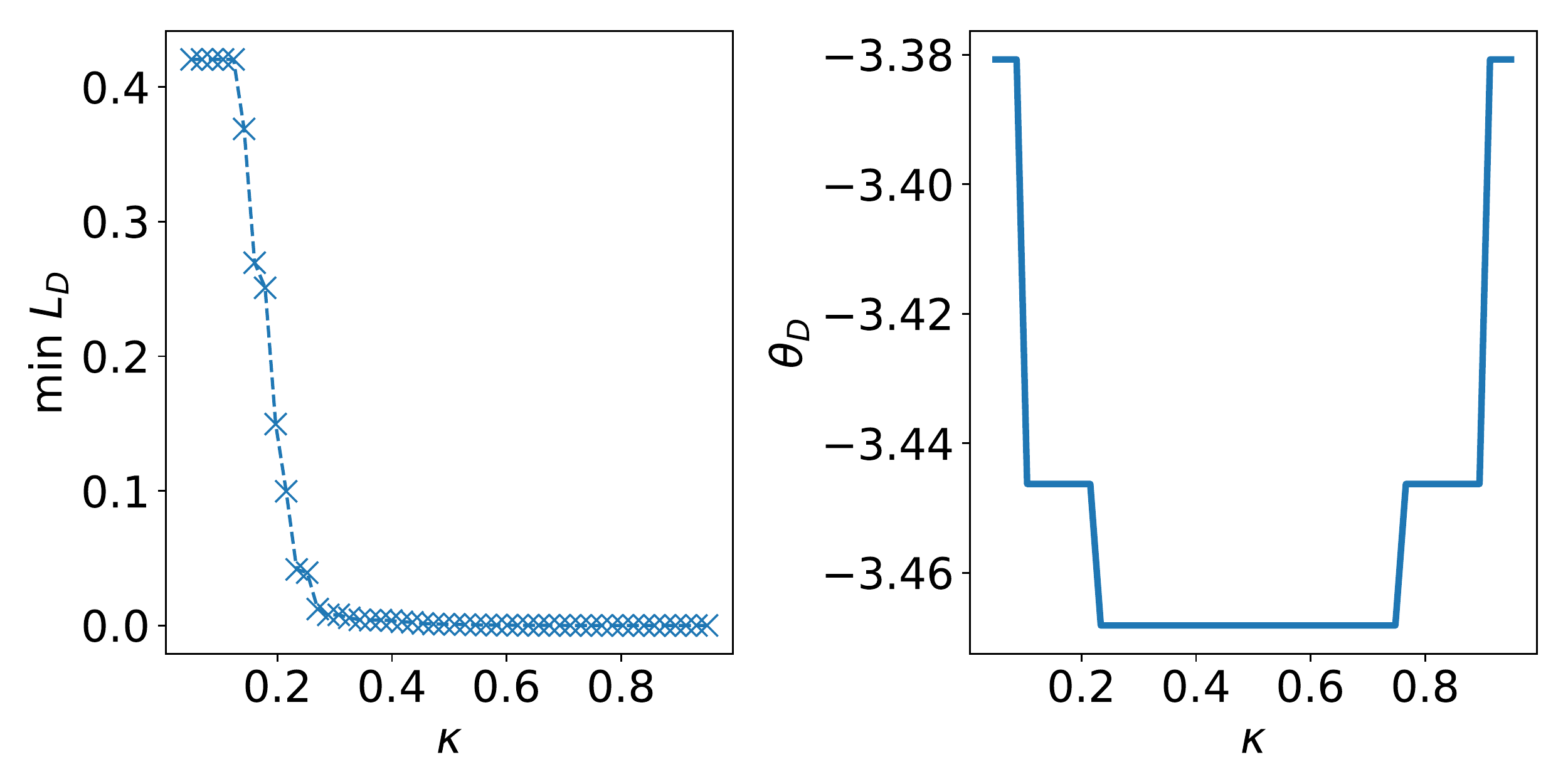}} 
    \end{tabular}
    \caption{The effect of $\kappa$. Comparison of theoretical predictions of minimum possible discriminator and generator losses to observed minimum losses when training DCGAN on CIFAR10. The left plots in each case (cross-dashed) show the experimental DCGAN results, and the right plots show the theoretical results $\vartheta_G, \vartheta_D$. $p=q=5$ and $\sigma_z=1$ are used in the theoretical calculations, to best match the DCGAN architecture.}
    \label{fig:vary_kappa_results}
\end{figure}

\section{Conclusions and outlook}\label{sec:conclusions}
We have contributed a novel model for the study of large neural network gradient descent dynamics with statistical physics techniques, namely an interacting spin-glass model for generative adversarial neural networks. We believe this is the first attempt in the literature to incorporate advanced architectural features of modern neural networks, beyond basic single network multi-layer perceptrons, into such statistical physics style models. We have conducted an asymptotic complexity analysis via Kac-Rice formulae and Random Matrix Theory calculations of the energy surface of this model, acting as a proxy for GAN training loss surfaces of large networks. Our analysis has revealed a banded critical point structure as seen previously for simpler models, explaining the surprising success of gradient descent in such complicated loss surfaces, but with added structural features that offer explanations for the greater difficulty of training GANs compared to single networks. We have used our model to study the effect of some elementary GAN hyper-parameters and compared with experiments training real GANs on a standard computer vision dataset. We believe that the interesting features of our model, and their correspondence with real GANs, are yet further compelling evidence for the role of statistical physics effects in deep learning and the value of studying such models as proxies for real deep learning models, and in particular the value of concocting more sophisticated models that reflect aspects of modern neural network design and practice. \\

From a mathematical perspective, we have extensively studied the limiting spectral density of a novel random matrix \ensemble using supersymmetric methods. In the preparation of this paper, we made considerable efforts to complete the average absolute value determinant calculations directly using a supersymmetric representation, as seen in \cite{baskerville2020loss}, however this was found to be analytically intractable (as expected), but also extremely troublesome numerically (essentially due to analytically intractable and highly complicated Riemann sheet structure in $\mathbb{C}^2$). We were able to sidestep these issues by instead using a Coulomb gas approximation, whose validity we have rigorously proved using a novel combination of concentration arguments and supersymmetric asymptotic expansions. We have verified with numerical simulations our derived mean spectral density for the relevant Random Matrix Theory  \ensemble and also the accuracy of the Coulomb gas approximation.\\

We hope that future work will be inspired to further study models of neural networks such as we have considered here. Practically, it would be exciting to explore the possibility of using our insights into GAN loss surfaces to devise algorithmic methods of avoiding training failure. Mathematically, the local spectral statistics of our random matrix \ensemble may be interesting to study, particularly around the cusp where the two disjoint components of the limiting spectral density merge.

\section{Acknowledgements}
FM is grateful for support from the University Research Fellowship of the University of Bristol. JPK is pleased to acknowledge support from a Royal Society Wolfson Research Merit Award and ERC Advanced Grant 740900 (LogCorRM). NPB is grateful to Diego Granziol for useful comments and discussions.

\printbibliography

\end{document}